\documentclass[12pt]{article}
\usepackage[margin=1in]{geometry}
\usepackage[small]{titlesec}
\usepackage{natbib,amsmath,amssymb,amsthm,setspace,paralist,rotating,booktabs,tabularx}
\AtBeginDocument{%
\abovedisplayskip=5pt plus 3pt minus 2pt
\belowdisplayskip=\abovedisplayskip
\abovedisplayshortskip=0pt plus 3pt
\belowdisplayshortskip=5pt plus 3pt minus 2pt}
\titlespacing*{\section}{0pt}{*2}{*1}
\titlespacing*{\subsection}{0pt}{*2}{*1}
\bibpunct{(}{)}{;}{a}{}{,}
\newtheorem{thm}{Theorem}
\newtheorem{lem}{Lemma}
\theoremstyle{definition}
\newtheorem{con}{Condition}
\newtheorem*{alg}{Coordinate Descent Algorithm}
\newcommand\rel[1]{\mathrel{\phantom{#1}}}
\DeclareMathOperator{\sgn}{sgn}
\DeclareMathOperator*\argmin{\arg\min}
\def\bb{\mathbf{b}}
\def\be{\mathbf{e}}
\def\bs{\mathbf{s}}
\def\bu{\mathbf{u}}
\def\bv{\mathbf{v}}
\def\bD{\mathbf{D}}
\def\bS{\mathbf{S}}
\def\bU{\mathbf{U}}
\def\bV{\mathbf{V}}
\def\bW{\mathbf{W}}
\def\bZ{\mathbf{Z}}
\def\bzero{\mathbf{0}}
\def\balpha{\boldsymbol{\alpha}}
\def\bbeta{\boldsymbol{\beta}}
\def\btheta{\boldsymbol{\theta}}
\def\bSigma{\boldsymbol{\Sigma}}
\def\ve{\varepsilon}

\def\mid{\,|\,}
\def\Bigmid{\,\Big|\,}

\allowdisplaybreaks
\begin{document}
\title{High-Dimensional Sparse Additive Hazards Regression}
\author{Wei Lin and Jinchi Lv}
\date{}
\maketitle
\footnotetext{Wei Lin is Postdoctoral Researcher, Department of Biostatistics and Epidemiology, Perelman School of Medicine, University of Pennsylvania, Philadelphia, PA 19104 (E-mail: \textit{weilin1@mail.med.upenn.edu}). Jinchi Lv is Assistant Professor, Information and Operations Management Department, Marshall School of Business, University of Southern California, Los Angeles, CA 90089 (E-mail: \textit{jinchilv@marshall.usc.edu}). This research was partially supported by NSF CAREER Award DMS-0955316 and Grant DMS-0806030. The authors sincerely thank the Co-Editor, Associate Editor, and two referees for their valuable comments that led to substantial improvement of the article.}

\begin{abstract}
High-dimensional sparse modeling with censored survival data is of great practical importance, as exemplified by modern applications in high-throughput genomic data analysis and credit risk analysis. In this article, we propose a class of regularization methods for simultaneous variable selection and estimation in the additive hazards model, by combining the nonconcave penalized likelihood approach and the pseudoscore method. In a high-dimensional setting where the dimensionality can grow fast, polynomially or nonpolynomially, with the sample size, we establish the weak oracle property and oracle property under mild, interpretable conditions, thus providing strong performance guarantees for the proposed methodology. Moreover, we show that the regularity conditions required by the $L_1$ method are substantially relaxed by a certain class of sparsity-inducing concave penalties. As a result, concave penalties such as the smoothly clipped absolute deviation (SCAD), minimax concave penalty (MCP), and smooth integration of counting and absolute deviation (SICA) can significantly improve on the $L_1$ method and yield sparser models with better prediction performance. We present a coordinate descent algorithm for efficient implementation and rigorously investigate its convergence properties. The practical utility and effectiveness of the proposed methods are demonstrated by simulation studies and a real data example.
\end{abstract}

\emph{Running title:} High-Dimensional Sparse Additive Hazards Regression


\section{Introduction}

Advances in experimental technologies in molecular biology during the past decade have brought in a wealth of biomedical data. For instance, DNA microarrays now can be used to measure the expression of tens of thousands of genes in a sample of cells or to identify hundreds of thousands of single nucleotide polymorphisms for an individual at the same time. Data of this kind pose tremendous challenges to effective statistical inference, since the number of features, $p$, is large compared to the number of observations, $n$, in which case many classical inference methods can easily fail or become inapplicable. Variable selection, a powerful tool for sparse modeling, is a fundamental task in high-dimensional regression problems, which aims to select only a small set of important variables from a huge number of features, in the hope of alleviating the overfitting problem in high dimensions and improving the predictive power and interpretability of the resulting model. See, for example, \citet{Fan:Lv:sele:2010} for a review of recent developments in high-dimensional variable selection.

When clinical data on patient survivals are also available, it would be informative to link high-dimensional biomedical data to survival outcomes. A number of efforts have recently been made in this direction. Regularization methods, which can yield sparse models and hence perform simultaneous variable selection and estimation, are particularly useful and have gained increasing popularity. Several regularization methods originally developed for linear regression have been adapted to survival models. For example, \citet{Tibs:lass:1997} and \citet{Fan:Li:vari:2002} extended the Lasso and nonconcave penalized likelihood, respectively, to the Cox model, while \citet{Zhan:Lu:adap:2007} and \citet{Zou:note:2008} developed weighted $L_1$ methods for the Cox model. \citet{Cai:Fan:Li:Zhou:vari:2005} were the first to study regularization methods for survival models in a framework with $p$ growing with $n$; they demonstrated that the nonconcave penalized pseudo-partial likelihood estimator for multivariate failure time data enjoys the oracle property when $p$ grows slowly with $n$. \citet{Anto:Fryz:Letu:2010} studied the Dantzig selector for the Cox model in a high-dimensional setting, but did not address the issue of model selection consistency. In addition to their classical applications for survival analysis in public health, survival models have been widely used to model time-to-event data for credit risk analysis in finance and economics \citep{Jarrow:cred:2009,Fan:Lv:Qi:spar:2011}. Identifying important risk factors and quantifying their contributions are crucial aims of these problems.

Variable selection techniques for survival data have also been extended beyond the Cox model. As a useful alternative to the Cox model, the additive hazards model assumes that the hazard function of a failure time $T$ conditional on a $p$-vector of possibly time-dependent covariates $\bZ(\cdot)$ takes the form
\begin{equation}\label{eq:model}
\lambda(t\mid\bZ)=\lambda_0(t)+\bbeta_0^T\bZ(t),
\end{equation}
where $\lambda_0(\cdot)$ is an unspecified baseline hazard function and $\bbeta_0$ is a $p$-vector of regression coefficients (\citealp[p.~74]{Cox:Oake:anal:1984}; \citealp[p.~182]{Bres:Day:stat:1987}; \citealp{Lin:Ying:semi:1994}). The additive models provide a characterization of the regression effects different than the multiplicative models and have some remarkable features that are not shared by the latter. In particular, model \eqref{eq:model} pertains to the risk difference, or excess risk, a measure that is especially relevant and informative in epidemiological and clinical studies. Variable selection in model \eqref{eq:model} has been studied by a number of authors; for example, \citet{Leng:Ma:path:2007} proposed a weighted $L_1$ approach and \citet{Mart:Sche:cova:2009} considered several regularization methods including the Lasso and Dantzig selector.

Despite the aforementioned developments, a rigorous high-dimensional theory that can provide strong performance guarantees for regularization-based variable selection methods in the survival setting is still lacking. Specifically, it is unclear how high dimensionality these methods can handle and what conditions are required for obtaining model selection consistency and nice sampling properties. The need for the development of such a theory is urgent, in view of the recent advances in understanding the performance of regularization methods in the linear regression and classification contexts \citep[e.g.,][]{Zhao:Yu:on:2006,Fan:Fan:fair:2008,Lv:Fan:unif:2009,Wain:shar:2009}.

In this article, we propose a class of regularization methods for simultaneous variable selection and estimation in model \eqref{eq:model}, by combining the nonconcave penalized likelihood approach \citep{Fan:Li:vari:2001} and the pseudoscore method \citep{Lin:Ying:semi:1994}. To justify the superior performance of the proposed methodology, we consider a high-dimensional setting where the dimension of covariates can grow fast, possibly nonpolynomially, with the sample size. Under mild, interpretable conditions, we establish the weak oracle property \citep{Lv:Fan:unif:2009} and oracle property \citep{Fan:Li:vari:2001} of the proposed regularized estimators. Our high-dimensional analysis is innovative in that it involves empirical process techniques that, to the best of our knowledge, have not been previously used in the survival analysis literature, and provides new insights into the model selection properties of regularization methods for survival data. In particular, we show that the regularity conditions required by the $L_1$ method are substantially relaxed by a certain class of sparsity-inducing concave penalties, which includes some commonly used concave penalties as special cases (see Section \ref{sec:pen}). Furthermore, we present a coordinate descent algorithm for efficient implementation and rigorously investigate its convergence properties. The practical utility and effectiveness of the proposed methodology are demonstrated by both simulated and real data.

In an independent work closely related to this article, \citet{Brad:Fan:Jian:2011} studied regularized estimation for variable selection in the Cox model and obtained important oracle-type theoretical results in which the dimension of covariates may grow nonpolynomially with the sample size. Besides model assumptions, a critical difference from our results, however, is that they imposed a \emph{random} condition on a large \emph{empirical} covariance matrix; see their condition 8. Thus, it is natural to ask the question whether the regularized estimators still enjoy the desired properties if a similar condition is imposed on the population version of the matrix. Since the empirical covariance matrix involves the outcome variables, as is generally the case for survival models, a \emph{nonrandom} condition on the \emph{population} covariance matrix seems more natural and will provide more confirmative performance guarantees. Such conditions will also have the benefit that they can be viewed as high-dimensional extensions of the classical asymptotic regularity conditions in the low-dimensional setting, which are imposed on the population covariance matrix. We will provide an affirmative answer to this important question.

The remainder of this article is organized as follows. In Section \ref{sec:method}, we propose a class of regularization methods and discuss choices of the penalty function. The theoretical properties of these regularized estimators are studied in Section \ref{sec:hd_theory}. In Section \ref{sec:hd_comp}, we describe a coordinate descent algorithm, study its convergence properties, and discuss selection of tuning parameters. Simulation studies and a real data example are presented in Sections \ref{sec:sim} and \ref{sec:data}, respectively. Some discussion is provided in Section \ref{sec:disc}, and all proofs and technical details are relegated to the Appendices.

\section{Regularization Methodology}\label{sec:method}

\subsection{Regularized Estimation}
We begin with the problem setup. Let $T$ be the failure time and $C$ the censoring time. Denote the censored failure time by $X=T\wedge C$ and the failure indicator by $\Delta=I(T\le C)$, where $I(\cdot)$ is the indicator function. Let $\bZ(\cdot)=(Z_1(\cdot),\cdots,Z_p(\cdot))$ be a vector of predictable covariate processes and assume that $T$ and $C$ are conditionally independent given $\bZ(\cdot)$. The observed data consist of $(X_i,\Delta_i,\bZ_i(\cdot))$, $i=1,\cdots,n$, which are independent copies of $(X,\Delta,\bZ(\cdot))$. We assume that the conditional hazard function is given by model \eqref{eq:model}.

We adopt the usual counting process notation. Define the observed-failure counting process $N_i(t)=I(X_i\le t,\Delta_i=1)$, the at-risk indicator $Y_i(t)=I(X_i\ge t)$, and the counting process martingale
\[
M_i(t)=N_i(t)-\int_0^tY_i(s)\{\lambda_0(s)+\bbeta_0^T\bZ_i(s)\}\,ds.
\]
We will also use $N(t)$, $Y(t)$, and $M(t)$ to denote the generic versions of these processes.

The pseudoscore function of \citet{Lin:Ying:semi:1994} is defined as
\[
\bU(\bbeta)=\frac{1}{n}\sum_{i=1}^n\int_0^{\tau}\{\bZ_i(t)-\overline{\bZ}(t)\}\{dN_i(t)-Y_i(t)\bbeta^T\bZ_i(t)\,dt\},
\]
where $\bbeta \in \mathbb{R}^p$, $\overline{\bZ}(t)=\sum_{j=1}^nY_j(t)\bZ_j(t)/\sum_{j=1}^nY_j(t)$, and $\tau$ is the maximum follow-up time. This estimating function is linear in $\bbeta$; through some algebraic manipulation, we can write $\bU(\bbeta)=\bb-\bV\bbeta$ with
\[
\bb=\frac{1}{n}\sum_{i=1}^n\int_0^{\tau}\{\bZ_i(t)-\overline{\bZ}(t)\}\,dN_i(t)
\]
and
\begin{equation}\label{eq:V_def}
\bV=\frac{1}{n}\sum_{i=1}^n\int_0^{\tau}Y_i(t)\{\bZ_i(t)-\overline{\bZ}(t)\}^{\otimes2}\,dt,
\end{equation}
where $\bv^{\otimes2}=\bv\bv^T$ for any vector $\bv$. Since $\bV$ is positive semidefinite, integrating $-\bU(\bbeta)$ with respect to $\bbeta$ leads to the least-squares-type loss function
\begin{equation}\label{eq:loss}
L(\bbeta)=\frac{1}{2}\bbeta^T\bV\bbeta-\bb^T\bbeta.
\end{equation}
Using this loss function for regularized estimation in model \eqref{eq:model} has been suggested by a number of authors, including \citet{Leng:Ma:path:2007} and \citet{Mart:Sche:cova:2009}; the latter authors noted also that $L(\bbeta)$ can be interpreted as an empirical prediction error, up to a constant, for the part of the model orthogonal to the at-risk indicator.

We now define the regularized estimator $\widehat{\bbeta}$ as a solution to the regularization problem
\begin{equation}\label{eq:opt}
\widehat{\bbeta}=\argmin_{\bbeta\in\mathbb{R}^p}\biggl\{Q(\bbeta)\equiv L(\bbeta)+\sum_{j=1}^pp_{\lambda}(|\beta_j|)\biggr\},
\end{equation}
where $\bbeta=(\beta_1,\cdots,\beta_p)^T$ and $p_{\lambda}(\theta)$, $\theta\ge0$, is a penalty function that depends on the regularization parameter $\lambda \geq 0$. When the minimization problem \eqref{eq:opt} is nonconvex, we will consider a local minimizer, as is common in the literature. It is often convenient to rewrite the penalty function as $p_{\lambda}(\cdot)=\lambda\rho_{\lambda}(\cdot)$; we write $\rho_{\lambda}(\cdot)$ as $\rho(\cdot)$ when it is free of $\lambda$. Without the penalty term, $\widehat{\bbeta}$ reduces to the pseudoscore estimator of \citet{Lin:Ying:semi:1994}. When the dimensionality is high, however, some form of regularization is needed to guard against overfitting, and the performance of the regularized estimator depends critically on the choice of the penalty function. Thus, in what follows we will first define a general class of penalty functions and discuss several popular choices among the class, and then present some theory to gain further insight into these choices.

\subsection{Penalty Function}\label{sec:pen}
To answer the question on what kind of penalty functions are ideal for model selection, \citet{Fan:Li:vari:2001} advocated penalty functions giving rise to estimators with three desired properties: sparsity, unbiasedness, and continuity. These properties motivate consideration of a class of penalty functions that satisfies the following condition.

\begin{con}\label{con:pen}
The function $\rho_{\lambda}(\theta)$ is increasing and concave in $\theta\in[0,\infty)$, and has a continuous derivative $\rho_{\lambda}'(\theta)$ on $(0,\infty)$. In addition, $\rho_{\lambda}'(\theta)$ is increasing in $\lambda$ and $\rho_{\lambda}'(0+)\equiv\rho'(0+)>0$ is independent of $\lambda$.
\end{con}

Some intuition for Condition \ref{con:pen} is as follows. The singularity at the origin encourages sparsity; the concavity assumption aims to reduce the estimation bias; the requirement that $\rho_{\lambda}'(\theta)$ is increasing in $\lambda$ allows $\lambda$ to effectively control the overall strength of the penalty. It should be noted that we do \emph{not} require \emph{strict} concavity or monotonicity, so that a wide range of penalty functions, including those that do not lead to all of the aforementioned three properties, are included in this class, which will facilitate our comparisons among different penalty functions. In the contexts of (generalized) linear models, this class of penalty functions has been studied by \citet{Lv:Fan:unif:2009} and \citet{Fan:Lv:nonc:2011}. Of particular interest are the following examples.

\begin{itemize}
\item The Lasso \citep{Tibs:regr:1996} uses the $L_1$-penalty, i.e., $\rho(\theta)=\theta$, $\theta\ge0$.
\item The smoothly clipped absolute deviation (SCAD) penalty \citep{Fan:comm:1997,Fan:Li:vari:2001} is given by the derivative
    \[
    \rho_{\lambda}'(\theta)=I(\theta\le\lambda)+\frac{(a\lambda-\theta)_+}{(a-1)\lambda}I(\theta>\lambda),\qquad\theta\ge0,
    \]
    where $a>2$ is a shape parameter. This penalty function takes off at the origin as the $L_1$-penalty and then gradually levels off until its derivative reaches zero.
\item The minimax concave penalty (MCP) proposed by \citet{Zhan:near:2010} has the derivative
    \[
    \rho_{\lambda}'(\theta)=\frac{(a\lambda-\theta)_+}{a\lambda},\qquad\theta\ge0,
    \]
    where $a>1$ is a shape parameter. In a similar spirit to SCAD, this penalty function gradually decreases its derivative to zero, except that it drops the $L_1$ part of SCAD.
\item The smooth integration of counting and absolute deviation (SICA) penalty \citep{Lv:Fan:unif:2009} takes the form
    \begin{equation}\label{eq:sica}
    \rho(\theta)=\frac{(a+1)\theta}{a+\theta},\qquad\theta\ge0,
    \end{equation}
    where $a>0$ is a shape parameter. With $a$ varying from 0 to $\infty$, this family provides a smooth homotopy between the $L_0$- and $L_1$-penalties. Each penalty function starts with slope $1+a^{-1}$ at the origin, passes through the point $(1,1)$, and decreases its slope toward zero over the interval $[0,\infty)$.
\end{itemize}

The $L_1$-penalty is a convex relaxation of the $L_0$-penalty and falls at the boundary of the class of penalty functions that satisfies Condition \ref{con:pen}. Although the $L_1$-regularization method enjoys the advantage of computational simplicity, it can suffer from several drawbacks which have motivated a number of improvements. The SCAD penalty was originally proposed to alleviate the bias caused by the $L_1$ approach, and has been shown to possess the oracle property, i.e., the resulting estimator performs asymptotically as well as the oracle estimator which knew the true sparse model in advance. The estimation bias of the Lasso can also interfere with variable selection; as a result, more stringent conditions such as the irrepresentable condition \citep{Zhao:Yu:on:2006} are typically required for consistent variable selection. The advantages of concave penalties regarding model selection consistency have recently been revealed and justified by a number of authors. \citet{Zhan:near:2010} showed that the MCP penalty has certain minimax optimality which enables it to strike a balance between the superior theoretical properties of concave penalties and the computational cost of nonconvex regularization problems. By investigating a nonasymptotic weak oracle property, \citet{Lv:Fan:unif:2009} showed that the regularity conditions needed for the $L_1$ approach can be substantially relaxed by using concave penalties. The SICA family proposed in that article has the remarkable feature that it can be used to define a sequence of regularization problems with varying theoretical performance and computational complexity.

\section{Theoretical Properties}\label{sec:hd_theory}

Besides the choice of the penalty function, the performance of the regularized estimators depends on a variety of factors, such as the dimensionality of the model, the dependency among the covariates, and the choice of the regularization parameter. In order to determine how these factors interact with each other and together affect the performance of the proposed estimators, in this section we rigourously develop a high-dimensional theory and discuss some of its implications.

We begin by introducing some notation to be used in our theoretical results. For any vector $\bv$, recall that $\bv^{\otimes2}=\bv\bv^T$ and for notational convenience, we write $\bv^{\otimes0}=1$ and $\bv^{\otimes1}=\bv$. Define
\begin{gather}
\bs^{(k)}(t)=E\{Y(t)\bZ(t)^{\otimes k}\},\qquad k=0,1,2,\label{eq:s_k}\\
\be(t)=\bs^{(1)}(t)/s^{(0)}(t),\notag\\
\bD=E\biggl[\int_0^{\tau}Y(t)\{\bZ(t)-\be(t)\}^{\otimes2}\,dt\biggr],\notag
\end{gather}
and
\[
\bSigma=E\biggl[\int_0^{\tau}\{\bZ(t)-\be(t)\}^{\otimes2}\,dN(t)\biggr].
\]
It is worthwhile to note that $\bD$ is the population counterpart of the matrix $\bV$ defined in \eqref{eq:V_def}, while $\bSigma$ is the population counterpart of the matrix
\[
\bW=\frac{1}{n}\sum_{i=1}^n\int_0^{\tau}\{\bZ_i(t)-\overline{\bZ}(t)\}^{\otimes2}\,dN_i(t).
\]
These matrices characterize the covariance structure of the model and will play a key role in our high-dimensional analysis.

Furthermore, define the \emph{active set} $A=\{j\colon\beta_{0j}\ne0\}$, where $\beta_{0j}$, $1\le j\le p$, is the $j$th component of the true regression coefficient vector $\bbeta_0$. Let $s=|A|$, i.e., the number of nonzero coefficients in $\bbeta_0$, and we allow the dimension triple $(s,n,p)$ to vary freely. Similarly, define the \emph{estimated active set} $\widehat{A}=\{j\colon\widehat{\beta}_j\ne0\}$, where $\widehat{\bbeta}=(\widehat{\beta}_1,\cdots,\widehat{\beta}_p)^T$. Denote the complement of a set $B$ by $B^c$. We will use sets to index vectors and matrices; for example, $\bbeta_{0A}$ is the vector formed by the components $\beta_{0j}$ of $\bbeta_0$ with $j\in A$, and $\bD_{A^cA}$ is the matrix formed by the entries $D_{ij}$ of $\bD$ with $i\in A^c$ and $j\in A$. Define the (half) \emph{minimum signal}
\[
d=\frac{1}{2}\min_{j\in A}|\beta_{0j}|.
\]
For any $\btheta=(\theta_1,\cdots,\theta_q)^T\in\mathbb{R}^q$ with $\theta_j\ne0$ for all $j$, following \cite{Lv:Fan:unif:2009}, define the \emph{local concavity} of the penalty function $\rho_{\lambda}(\cdot)$ at point $\btheta$ as
\[
\kappa(\rho_\lambda;\btheta)=\lim_{\ve\to0+}\max_{1\le j\le q}\sup_{|\theta_j|-\ve<t_1<t_2<|\theta_j|+\ve}\biggl\{-\frac{\rho_{\lambda}'(t_2)-\rho_{\lambda}'(t_1)} {t_2-t_1}\biggr\}.
\]
Finally, define
\begin{gather*}
\kappa_0=\sup\{\kappa(\rho_{\lambda};\btheta)\colon\btheta\in\mathbb{R}^s,\|\btheta-\bbeta_{0A}\|_{\infty}\le d\},\\
\varphi=\|\bD_{AA}^{-1}\|_{\infty},
\end{gather*}
and
\[
\mu=\Lambda_{\min}(\bD_{AA})-\lambda\kappa_0,
\]
where $\Lambda_{\min}(\cdot)$ denotes the minimum eigenvalue. It is important to note that all the quantities defined above can depend on the sample size $n$, and we have suppressed this dependency for notational simplicity.

\subsection{Weak Oracle Property}
\citet{Lv:Fan:unif:2009} introduced the concept of weak oracle property for comparing different regularization methods. An estimator is said to have the weak oracle property if it is both consistent in model selection and uniformly consistent in estimation. This notion is weaker than the oracle property introduced by \citet{Fan:Li:vari:2001} and hence can be satisfied by a broader class of estimators. To derive a nonasymptotic result regarding the weak oracle property of the proposed estimators, we need to impose the following conditions.

\begin{con}\label{con:proc}
\begin{inparaenum}[(i)]
\item $\int_0^{\tau}\lambda_0(t)\,dt<\infty$.
\item $P\{Y(\tau)=1\}>0$.
\item There exist constants $D,K,r>0$ such that
\[
P\biggl(\sup_{t\in[0,\tau]}|Z_j(t)|>x\biggr)\le D\exp(-Kx^r)
\]
for all $x>0$ and $j=1,\cdots,p$.
\item The sample paths of $Z_j(\cdot)$, $j=1,\cdots,p$, are of uniformly bounded variation.
\end{inparaenum}
\end{con}

\begin{con}\label{con:inco}
There exist constants $\alpha\in(0,1]$, $\gamma\in[0,1/2]$, and $c>0$ such that
\[
\|\bD_{A^cA}\bD_{AA}^{-1}\|_{\infty}\le\biggl\{(1-\alpha)\frac{\rho'(0+)}{\rho_{\lambda}'(d)}\biggr\}\wedge(cn^{\gamma}).
\]
\end{con}

In Condition \ref{con:proc}, parts (i) and (ii) are standard for survival models; part (iii) controls the tail behavior of the covariates and is trivially satisfied for bounded covariates; part (iv) is a very mild technical condition that will facilitate entropy calculations.

Condition \ref{con:inco} is an analog of condition (35) in \citet{Lv:Fan:unif:2009} for penalized least squares, which is in turn a generalization of condition (15) in \citet{Wain:shar:2009} for the Lasso. Often for linear regression, such conditions are first imposed on the deterministic Gram matrix, and then a variety of random design matrices such as Gaussian ensembles can be further considered. For survival models such as model \eqref{eq:model}, however, there is no exact analog of the deterministic Gram matrix; here the matrix $\bV$, which plays the same role as the Gram matrix in linear regression, involves the at-risk indicators and hence is \emph{nondeterministic}. Thus, Condition \ref{con:inco} is imposed on the population version $\bD$ of $\bV$. Note also that we are not restricted to the cases where the covariates are bounded or Gaussian.

The right-hand side of Condition \ref{con:inco} consists of two parts: the first part is an upper bound that reflects the intrinsic capability of the penalty function for variable selection; the second part is at most $O(\sqrt{n})$, where the parameter $\gamma$ needs to be determined by other conditions to be presented later. For the $L_1$-penalty, the first part is bounded by constant 1, which is stringent; for concave penalties, the upper bound is generally relaxed, since concavity implies that $\rho_{\lambda}'(\theta)$ is decreasing in $\theta$ and thus $\rho'(0+)/\rho_{\lambda}'(d)$ can diverge asymptotically. When signals are fairly strong so that $d\gg\lambda$, the first part imposes no constraint for the SCAD and MCP penalties, since $\rho_{\lambda}'(d)=0$ in that case. Also, the upper bound for the SICA penalty can be substantially relaxed by choosing a small value of $a$.

Since Condition \ref{con:inco} and definitions of $\varphi$ and $\mu$ involve the matrices $\bD_{A^cA}\bD_{AA}^{-1}$, $\bD_{AA}^{-1}$, and $\bD_{AA}$, a key step to establishing the weak oracle property is to show that the empirical counterparts of these matrices are close to them in some sense. This intermediate result is provided by the following lemma, which gives explicit probability bounds for similar conditions to hold for the empirical matrices. In what follows, let $\Omega_L$ denote the event that $\max_{j=1}^p\sup_{t\in[0,\tau]}|Z_j(t)|\le L$ for $L>0$.

\begin{lem}[Concentration of empirical matrices]\label{lem:sample1}
Under Conditions \ref{con:pen}--\ref{con:inco}, if $\mu>0$ and $\varphi\vee\mu^{-1}=O(\sqrt{n}/s)$, then there exist constants $D,K>0$ such that
\begin{gather}
P(\|\bV_{AA}^{-1}\|_{\infty}\ge2\varphi\mid\Omega_L)\le s^2D\exp\biggl\{-K\frac{n}{L^4}\biggl(\frac{1}{\varphi^2s^2}\wedge1\biggr)\biggr\},\label{eq:phi}\\
\begin{aligned}
&P\biggl[\|\bV_{A^cA}\bV_{AA}^{-1}\|_{\infty}\ge\biggl\{\Bigl(1-\frac{\alpha}{2}\Bigr)\frac{\rho'(0+)}{\rho_{\lambda}'(d)}\biggr\}\wedge(2cn^{\gamma})\mid\Omega_L\biggr]\\
&\qquad\le(p-s)sD\exp\biggl[-K\frac{n}{L^4}\biggl\{\frac{\bigl(\rho_{\lambda}'(d)^{-1}\wedge n^{\gamma}\bigr)^2}{\varphi^2s^2}\wedge1\biggr\}\biggr]\\
&\qquad\rel{\le}{}+s^2D\exp\biggl\{-K\frac{n}{L^4}\biggl(\frac{1}{\varphi^2s^2}\wedge1\biggr)\biggr\},
\end{aligned}\label{eq:rho}
\end{gather}
and
\begin{equation}\label{eq:kappa}
P(\Lambda_{\min}(\bV_{AA})\le\lambda\kappa_0\mid\Omega_L)\le s^2D\exp\biggl\{-K\frac{n}{L^4}\biggl(\frac{\mu^2}{s^2}\wedge1\biggr)\biggr\}.
\end{equation}
\end{lem}

Inequalities \eqref{eq:phi} and \eqref{eq:rho} show that there would not be much difference if we had defined the quantity $\varphi$ or imposed Condition \ref{con:inco} on the empirical matrices. The eigenvalue condition $\Lambda_{\min}(\bV_{AA})>\lambda\kappa_0$ is needed for identification of a strict local minimizer of problem \eqref{eq:opt}; inequality \eqref{eq:kappa} says that this condition holds with high probability if $\Lambda_{\min}(\bD_{AA})$ and $\lambda\kappa_0$ have a positive gap $\mu$ that does not shrink to zero too fast.

We now state our main theoretical result regarding the weak oracle property of the proposed estimators.

\begin{thm}[Weak oracle property]\label{thm:weak_oracle}
In addition to Conditions \ref{con:pen}--\ref{con:inco}, assume that the following conditions hold:
\begin{gather}
\frac{n\bigl(\rho_{\lambda}'(d)^{-1}\wedge n^{\gamma}\bigr)^2}{\varphi^2s^2(\log p)^{r_1}}\to\infty,\qquad\frac{n(\varphi^{-1}\wedge\mu)^2}{s^2(\log s)^{r_1}}\to\infty,\label{eq:wop_dim}\\
\frac{n\lambda^2}{(\log p)^{r_1}}\to\infty,\qquad\frac{n^{1-2\gamma}\lambda^2}{(\log s)^{r_1}}\to\infty,\label{eq:wop_lam}
\end{gather}
and
\begin{equation}\label{eq:wop_d}
d\ge c_1\varphi\lambda\rho'(0+),
\end{equation}
where $\mu>0$, $r_1=(r+4)/r$, and $c_1=2+1/(4c)$. Then, for some constants $D,K>0$, with probability at least
\[
1-D\exp\biggl[-Kn^{1/r_1}\biggl\{\frac{(\varphi^{-1}\wedge\mu)^2}{s^2}\wedge1\biggr\}^{1/r_1}\biggr]
-D\exp\biggl\{-Kn^{1/r_1}\biggl(\frac{\lambda^2}{n^{2\gamma}}\wedge1\biggr)^{1/r_1}\biggr\}\to1,
\]
there exists a regularized estimator $\widehat{\bbeta}$ that satisfies the following properties:
\begin{compactenum}[(a)]
\item (Sparsity) $\widehat{\bbeta}_{A^c}=\bzero$.
\item ($L_{\infty}$-loss) $\|\widehat{\bbeta}_A-\bbeta_{0A}\|_{\infty}\le c_1\varphi\lambda\rho'(0+)$.
\end{compactenum}
\end{thm}

To develop intuition for the two conditions in \eqref{eq:wop_dim}, we consider some simplified cases. First, concavity implies that $\rho_{\lambda}'(d)\le\rho'(0+)$; thus, a sufficient condition for the first condition in \eqref{eq:wop_dim} to hold is
\begin{equation}\label{eq:wop_dim1}
\frac{n}{\varphi^2s^2(\log p)^{r_1}}\to\infty.
\end{equation}
Consider the second condition in \eqref{eq:wop_dim} and recall that $\mu=\Lambda_{\min}(\bD_{AA})-\lambda\kappa_0$. For the $L_1$-penalty, $\kappa_0=0$; for SCAD and MCP, $\lambda\kappa_0=(a-1)^{-1}$ and $a^{-1}$, respectively. Thus, for this condition to hold for these penalties, it suffices to assume that $\Lambda_{\min}(\bD_{AA})$ is bounded away from zero and that
\[
\frac{n}{\varphi^2s^2(\log s)^{r_1}}\to\infty,
\]
where the latter is implied by \eqref{eq:wop_dim1}. Therefore, conditions in \eqref{eq:wop_dim} are primarily constraints on the growth rates of the model dimensions $p$ and $s$ and certain matrix norms of $\bD_{AA}^{-1}$. On the other hand, if we assume, for simplicity, that $\varphi$ is constant, then \eqref{eq:wop_dim1} gives a lower bound for the number of observations that are needed for guaranteed sparse recovery, $n\gg s^2(\log p)^{r_1}$. This is an interesting setting, since it shows that the proposed estimators can handle a nonpolynomially growing dimension of covariates as high as $\log p=o(n^{1/r_1})$, while the dimension of the true sparse model grows as $s=o(\sqrt{n})$. In particular, for bounded covariates, we can take $r_1=1$ by letting $r \rightarrow \infty$ and thus allow $\log p=o(n)$.

For simplicity, consider the case of bounded covariates, i.e., $r_1=1$. The two conditions in \eqref{eq:wop_lam} give a lower bound for the regularization parameter $\lambda$,
\[
\lambda\gg\sqrt{\frac{\log p}{n}}\vee\sqrt{\frac{\log s}{n^{1-2\gamma}}}.
\]
Thus, in view of \eqref{eq:wop_d}, we see that $\lambda$ should be chosen to satisfy
\[
\sqrt{\frac{\log p}{n}}\vee\sqrt{\frac{\log s}{n^{1-2\gamma}}}\ll\lambda\le\frac{d}{c_1\varphi\rho'(0+)}.
\]
For such choices of $\lambda$ to exist, the minimum signal $d$ must satisfy
\begin{equation}\label{eq:dgg}
d\gg\varphi\biggl(\sqrt{\frac{\log p}{n}}\vee\sqrt{\frac{\log s}{n^{1-2\gamma}}}\biggr).
\end{equation}
Recall that $\gamma\in[0,1/2]$ has appeared in Condition \ref{con:inco}. More insight can be gained by comparing the two parts on the right-hand side of \eqref{eq:dgg}: the first part will dominate if $n^{\gamma}\ll\sqrt{(\log p)/(\log s)}$, and in this case, Theorem \ref{thm:weak_oracle} guarantees recovery of signals that satisfy $d\gg\varphi\sqrt{(\log p)/n}$, independent of $\gamma$; otherwise, the second part will dominate, and the weakest recoverable signal will depend on the correlation between the two sets of true variables and noise variables as reflected by the value of $\gamma$. Of course, for the $L_1$-penalty, since the first part in Condition \ref{con:inco} always dominates the second part, we can simply take $\gamma=0$, and thus the value of $\gamma$ plays no role in determining the lower bound for $d$.

Despite the similarities between the results presented here and those for (generalized) linear models in \citet{Lv:Fan:unif:2009} and \citet{Fan:Lv:nonc:2011}, it is worthwhile to note some important differences. The restriction on the correlation structure described in Condition \ref{con:inco} has a more complex form, in that the matrix $\bD$ involves not only the covariates but also the failure process and censoring mechanism. This complexity arises from the semiparametric nature of survival models and the presence of censoring, and as a consequence, the ability of performing variable selection is affected by the interplay of several factors, including the covariates, baseline hazard, and censoring mechanism. Moreover, although our results allow the dimensions to grow at rates comparable to those available for linear regression models, our proofs suggest that the necessary sample size for observing the effects of these growth rates, which is determined by the constants, may be significantly larger. In addition, there is a rate loss in characterizing the convergence of the matrix $\bV$ to its population counterpart $\bD$. These facts call for a need to increase the sample size for making reliable inference.

\subsection{Oracle Property}\label{sec:oracle}
In addition to model selection consistency, the oracle property requires the regularized estimator to be asymptotically as efficient as the oracle estimator with the true sparse model known \emph{a priori}. For this purpose, some extra eigenvalue conditions are needed. Define $\Lambda_1=\Lambda_{\min}(\bD_{AA})$, $\Lambda_2=\Lambda_{\min}(\bSigma_{AA})$, and $\Lambda_3=\Lambda_{\min}(\bD_{AA}^{-1}\bSigma_{AA}\bD_{AA}^{-1})$. The oracle property of the proposed regularized estimators is stated in the following result.

\begin{thm}[Oracle property]\label{thm:oracle}
Assume that all conditions of Theorem \ref{thm:weak_oracle} hold. In addition, assume that
\begin{equation}\label{eq:op_dim}
\frac{n\Lambda_1^2}{s^2(\log s)^{r_1}}\to\infty,\qquad\frac{n\Lambda_2^2}{s^2}\to\infty,\qquad \frac{n\Lambda_1^4\Lambda_3}{s^3}\to\infty,
\end{equation}
and
\begin{equation}\label{eq:op_lam}
\frac{ns\lambda^2\rho_{\lambda}'(d)^2}{\Lambda_1^2\Lambda_3}\to0,
\end{equation}
where $r_1=(r+4)/r$. Then, for some constants $D,K>0$, with probability at least
\[
1-D\exp\biggl[-Kn^{1/r_1}\biggl\{\frac{(\varphi^{-1}\wedge\mu\wedge\Lambda_1)^2}{s^2}\wedge1\biggr\}^{1/r_1}\biggr]
-D\exp\biggl\{-Kn^{1/r_1}\biggl(\frac{\lambda^2}{n^{2\gamma}}\wedge1\biggr)^{1/r_1}\biggr\}\to1,
\]
there exists a regularized estimator $\widehat{\bbeta}$ that satisfies the following properties:
\begin{compactenum}[(a)]
\item (Sparsity) $\widehat{\bbeta}_{A^c}=\bzero$.
\item (Asymptotic normality) For every $\bu\in\mathbb{R}^s$ with $\|\bu\|_2=1$,
\[
\sqrt{n}\bu^T\bSigma_{AA}^{-1/2}\bD_{AA}(\widehat{\bbeta}_A-\bbeta_{0A})\xrightarrow{\mathcal{D}}N(0,1).
\]
\end{compactenum}
\end{thm}

The three conditions in \eqref{eq:op_dim} relate the sparsity dimension $s$ to eigenvalue bounds for the matrices $\bD_{AA}$, $\bSigma_{AA}$, and $\bD_{AA}^{-1}\bSigma_{AA}\bD_{AA}^{-1}$. If we consider the special case where the eigenvalues of these matrices are all bounded away from zero, then these conditions are trivially satisfied for $s=o(n^{1/3})$. The form of \eqref{eq:op_dim}, however, deals with much more general situations and is very meaningful in that the eigenvalue bounds are closely related to the difficulty of the estimation problem.

In the context of linear regression, it is well known that the $L_1$-penalty does not have the oracle property \citep{Zou:adap:2006,Wain:shar:2009}. It is clear from \eqref{eq:op_lam} that this is also the case for the problem considered here. To see this, consider the case where $\Lambda_1$ and $\Lambda_3$ are fixed, and note that $\rho'(d)\equiv1$ for the $L_1$-penalty. Conditions in \eqref{eq:wop_lam} imply that $n\lambda^2\to\infty$, and hence \eqref{eq:op_lam} cannot hold. For the SCAD and MCP penalties, \eqref{eq:op_lam} is trivially satisfied for $d\gg\lambda$, since $\rho_{\lambda}'(d)=0$ in that case. For the SICA penalty, we have $\rho'(d)=a(a+1)/(a+d)^2$; thus, in order to obtain the oracle property, we should take $a\to0+$ at a rate such that $d\gg a$ and $\sqrt{ns}\lambda a/d^2\to0$. This result is reasonable since the SICA penalty approaches the $L_0$-penalty as $a\to0+$.

\section{Implementation}\label{sec:hd_comp}

In this section, we describe an efficient coordinate descent algorithm for implementation of the proposed methodology, analyze its convergence properties, and discuss the selection of tuning parameters.

\subsection{Coordinate Descent Algorithm}
The idea of coordinate optimization for penalized least-squares problems was proposed by \citet{Fu:pena:1998} and \cite{Daub:Defr:De:2004}, and was demonstrated by \citet{Frie:Hast:Hfli:Tibs:path:2007} and \cite{Wu:Lang:coor:2008} to be exceptionally efficient for large-scale sparse problems. Recently, a number of authors, including \citet{Fan:Lv:nonc:2011}, \citet{Breh:Huan:coor:2011}, and \citet{Mazu:Frie:Hast:spar:2011}, generalized this idea to regularized regression with concave penalties and showed that it is an attractive alternative to earlier proposals such as the local quadratic approximation (LQA) \citep{Fan:Li:vari:2001} and local linear approximation (LLA) \citep{Zou:Li:one-:2008}.

To balance the regularization strengths on different components of $\bbeta$, we minimize the weighted version of the objective function,
\[
\widetilde{Q}(\bbeta;\lambda)\equiv L(\bbeta)+\sum_{j=1}^pV_{jj}p_{\lambda}(|\beta_j|),
\]
where, for simplicity, we assume that $V_{jj}\ne0$ for all $j$. The coordinate descent method minimizes the objective function in one coordinate at a time and cycles through all coordinates until convergence. To produce a solution path over a grid of values of the regularization parameter $\lambda$, at each grid point the solution from a neighboring point is used as a warm start to accelerate convergence and, for concave penalties, to avoid suboptimal local solutions. Pick $\lambda_{\max}$ sufficiently large such that $\widehat{\bbeta}=\bzero$ is a (local) solution and take a decreasing sequence $(\lambda_1,\cdots,\lambda_K)$ with $\lambda_1=\lambda_{\max}$; for the Lasso, SCAD, MCP, and SICA, one can take $\lambda_{\max}=\max_{j=1}^p(|b_j|/V_{jj})$, where $b_j$ is the $j$th component of $\bb$. The following algorithm produces a solution path $\widehat{\bbeta}^{\lambda_k}$ over a grid of points $\lambda_k$, $k=1,\cdots,K$.

\begin{alg}\label{alg:cd}\mbox{}
\begin{compactenum}
\item Initialize $\widehat{\bbeta}=\bzero$ and set $k=1$.
\item Cyclically for $j=1,\cdots,p$, update the $j$th component $\widehat\beta_j$ of $\widehat{\bbeta}$ by the univariate (global) minimizer of $\widetilde{Q}(\widehat{\bbeta};\lambda_k)$ with respect to $\widehat{\beta}_j$ until convergence.
\item Set $\widehat{\bbeta}^{\lambda_k}=\widehat{\bbeta}$ and $k\leftarrow k+1$.
\item Repeat Steps 2 and 3 until $k=K+1$.
\end{compactenum}
\end{alg}

The above algorithm requires efficiently solving a univariate regularization problem in Step~2. Since $L(\bbeta)$ is quadratic, closed-form solutions to this univariate problem exist for commonly used penalty functions, including the Lasso, SCAD, MCP, and SICA. The solutions for the first three penalties have been given in the aforementioned references, and we provide the solution for the SICA penalty in Appendix \ref{sec:app_sica}. Coordinate descent may be slow if the resulting model is not sparse; hence, to save computation time, one can set a level of sparsity for early stopping if only models up to a certain size are desired.

\subsection{Convergence Analysis}
Denote by $\bbeta^m=(\beta_1^m,\cdots,\beta_p^m)^T$ the $m$th update vector generated by coordinate descent, i.e.,
\[
\beta_j^m=\argmin_{\theta\in\mathbb{R}}\widetilde{Q}(\beta_1^m,\cdots,\beta_{j-1}^m,\theta,\beta_{j+1}^{m-1},\cdots,\beta_p^{m-1})
\]
for $j=1,\cdots,p$ and $m=1,2,\cdots$, where $\bbeta^0=(\beta_1^0,\cdots,\beta_p^0)^T$ is the starting point and we have suppressed the dependence of $\widetilde{Q}(\bbeta)$ on $\lambda$. Define the \emph{maximum concavity} of the penalty function $p_{\lambda}(\cdot)$ by
\[
\kappa(p_{\lambda})=\sup_{0<t_1<t_2<\infty}\biggl\{-\frac{p_{\lambda}'(t_2)-p_{\lambda}'(t_1)}{t_2-t_1}\biggr\}.
\]
For the $L_1$-penalty, SCAD, MCP, and SICA, we have $\kappa(p_{\lambda})=0$, $(a-1)^{-1}$, $a^{-1}$, and $2\lambda(a^{-1}+a^{-2})$, respectively. Due to the concavity of $p_{\lambda}(\cdot)$, the objective function $\widetilde{Q}(\bbeta)$ is generally nonconvex. However, under certain conditions such that the concavity of $p_{\lambda}(\cdot)$ is dominated by the convexity of the quadratic loss function $L(\bbeta)$ componentwise (subject to normalization by $V_{jj}$), $\widetilde{Q}(\bbeta)$ can still be strictly convex along each coordinate. This key observation leads to the following convergence result.

\begin{thm}[Convergence of coordinate descent]\label{thm:conv}
Under Condition \ref{con:pen}, the sequence $\{\bbeta^m\}$ generated by the coordinate descent algorithm is bounded. Moreover, the following statements hold:
\begin{compactenum}[(a)]
\item If the penalty function $p_{\lambda}(\cdot)$ satisfies $\kappa(p_{\lambda})<1$, then every cluster point of $\{\bbeta^m\}$ is a local minimizer of $\widetilde{Q}(\bbeta)$.
\item If in addition, the sequence $\{\bbeta^m\}$ eventually lies in a compact neighborhood $\mathcal{K}$ of $\bbeta^*$ such that $\bbeta^*$ is the unique local minimizer of $\widetilde{Q}(\bbeta)$ in $\mathcal{K}$, then the sequence $\{\bbeta^m\}$ converges to $\bbeta^*$.
\end{compactenum}
\end{thm}

Note that the condition $\kappa(p_{\lambda})<1$ is always satisfied by the $L_1$-penalty, SCAD ($a>2$), and MCP ($a>1$). For the SICA penalty, the condition is satisfied if $\lambda(a^{-1}+a^{-2})<1/2$. This latter condition suggests that no matter what value $\lambda$ takes, one can adjust the value of $a$ to ensure computational stability. Since the optimal $\lambda$ is often small, one can set $a$ to a small value to achieve good performance of the SICA method.

It is often useful to determine when a local minimizer $\widehat{\bbeta}$ of $\widetilde{Q}(\bbeta)$ is also a global minimizer and when the sequence $\{\bbeta^m\}$ converges to it. The (restricted) global optimality of nonconcave penalized likelihood estimators was characterized by \cite{Fan:Lv:nonc:2011}. The following result gives sufficient conditions for the (restricted) global optimality of $\widehat{\bbeta}$ on some subspace of $\mathbb{R}^p$ and the convergence of $\{\bbeta^m\}$ to the global optimum.

\begin{thm}[Restricted global optimality]\label{thm:global}
Let $\widehat{\bbeta}$ be a local minimizer of $\widetilde{Q}(\bbeta)$. For any subset $S$ of $\{1,\cdots,p\}$, denote by $\mathcal{B}_S$ the $|S|$-dimensional subspace $\{\bbeta\in\mathbb{R}^p\colon\beta_j=0,j\not\in S\}$. Under Condition \ref{con:pen}, the following statements hold:
\begin{compactenum}[(a)]
\item If $\widehat\bbeta$ lies in $\mathcal{B}_S$ and $\Lambda_{\min}(\bV_{SS})\ge\kappa(p_{\lambda})\max_{j\in S}V_{jj}$, then $\widehat{\bbeta}$ is a global minimizer of $\widetilde{Q}(\bbeta)$ in $\mathcal{B}_S$.
\item If the sequence $\{\bbeta^m\}$ generated by the coordinate descent algorithm eventually lies in $\mathcal{B}_S$ and $\Lambda_{\min}(\bV_{SS})>\kappa(p_{\lambda})\max_{j\in S}V_{jj}$, then $\widetilde{Q}(\bbeta)$ has a unique global minimizer $\bbeta^*$ in $\mathcal{B}_S$ and the sequence $\{\bbeta^m\}$ converges to $\bbeta^*$.
\end{compactenum}
\end{thm}

The condition in part (a) is trivially satisfied for the $L_1$-penalty. For the SCAD, MCP, and SICA, the condition can be satisfied with some $S$ if the correlation among covariates is not too strong and the concavity of the penalty function is not too large. Following \cite{Fan:Lv:nonc:2011}, under some mild regularity conditions we can further establish the global optimality of $\widehat{\bbeta}$ on the union of all $|S|$-dimensional coordinate subspaces of $\mathbb{R}^p$. Although the global optimality is somewhat hard to guarantee for concave penalties, we remark that it is not required for achieving practically good performance of the regularized estimator. In fact, our theoretical arguments suggest that a sparse local solution should possess nice statistical properties, while the coordinate descent algorithm is likely to find such a solution by carefully following a solution path from the sparsest end.

\subsection{Selection of Tuning Parameters}
After a solution path has been produced, the optimal regularization parameter $\lambda$ can be chosen by $M$-fold cross-validation. The cross-validation score is defined as
\[
\operatorname{CV}(\lambda)=\frac{1}{M}\sum_{m=1}^ML^{(m)}(\widehat{\bbeta}^{(-m)}(\lambda)),
\]
where $L^{(m)}(\cdot)$ is the loss function given in \eqref{eq:loss} computed from the $m$th part of the data, and $\widehat{\bbeta}^{(-m)}(\lambda)$ is the estimate from the data with the $m$th part removed. The concave penalties have one additional parameter $a$ to be tuned. For the SCAD penalty, $a=3.7$ was suggested by \citet{Fan:Li:vari:2001} from a Bayesian perspective and is commonly used in the literature. Since the MCP is similar to SCAD, we also take $a=3.7$. The choice of $a$ for the SICA penalty requires a little more caution, since a small $a$ that is often needed to yield a superior theoretical performance can sometimes cause the problem of computational instability. This can be clearly seen from the discussion following Theorem~\ref{thm:conv}; a too small $a$ would cause the convexity condition $\lambda(a^{-1}+a^{-2})<1/2$ to be violated. To solve this dilemma, we first compute a pilot solution path with a larger $a$ suggested by the convexity condition, and then set $a$ to a smaller value and recompute the solution path by taking the pilot solutions as warm starts. If needed, the above process can be repeated several times to gain further stability. Often we take $a=1$ for the pilot solution and one or a few values toward zero for recomputing the solution. We find that, in the settings we have considered, the practical performance of our methods is not sensitive to these choices. In the more difficult settings and if computing resources allow, a fine tuning on $a$ may yield additional benefits; see \citet{Mazu:Frie:Hast:spar:2011} for more discussion on producing a solution surface along $(\lambda,a)$.

\section{Simulation Studies}\label{sec:sim}

We conducted three simulation studies to evaluate the finite-sample performance of the proposed regularization methods with the Lasso, SCAD, MCP, and SICA penalties, and compare them to the oracle estimator which knew the true sparse model in advance. Since the elastic net (Enet) \citep{Zou:Hast:regu:2005} is also a common choice for the penalty function and is known to outperform the Lasso when the important variables are highly correlated, we also include it in our comparisons.

The first simulation study aims to examine the case where $p$ is comparable to but smaller than $n$. We generated data from the model
\[
\lambda(t\mid\bZ)=1+\bbeta_0^T\bZ,
\]
where $\bZ$ has a multivariate normal distribution with mean zero and covariance matrix $(\rho^{|i-j|})_{i,j=1}^p$ and subject to $\bbeta_0^T\bZ>-1$, and $\bbeta_0=(\bv^T,\cdots,\bv^T,\bzero)^T$ with the pattern $\bv=(1,0,-1,0,0,0)^T$ repeated $q$ times. We set $\rho=0.1$ and $0.5$, and $q=3$ so that the sparsity dimension $s=6$. We considered $p=50$ and 100, and $n=200$. The censoring time $C$ has a uniform distribution $U(0,c_0)$ with $c_0$ chosen to obtain a censoring rate about $25\%$.

We first choose the optimal regularization parameter $\lambda$ by tenfold cross-validation and evaluate the performance of the resulting estimators by six measures. The first two measures quantify prediction performance: PE1 is the prediction error based on the loss function $L(\cdot)$ (up to a constant), and PE2 is the $L_2$ prediction error in the excess risk, $\|\bZ^T(\widehat{\bbeta}-\bbeta_0)\|_2$, both computed from an independent test sample of size 500. For estimation accuracy, we report the $L_2$-loss $\|\widehat{\bbeta}-\bbeta_0\|_2$ and $L_1$-loss $\|\widehat{\bbeta}-\bbeta_0\|_1$. The other two measures pertain to model selection consistency: \#S and \#FN refer to the number of selected variables and the number of incorrectly excluded variables (false negatives), respectively. The means and standard deviations of each measure over 100 replicates are summarized in Table~\ref{tab:low}.

\begin{center}
[Table \ref{tab:low} about here.]
\end{center}

From Table~\ref{tab:low}, we see that the Lasso and Enet had very close performance in this setting and the concave penalties outperformed the Lasso and Enet in that they selected a sparser model with better prediction and estimation performance. As expected from their similarity, the SCAD and MCP had comparable performance, of which the latter selected a slightly sparser model than the former. The SICA performed best among the five methods, with a performance very close to that of the oracle estimator; in most cases, it was able to identify exactly the six important variables.

It has been noted that prediction-based criteria such as cross-validation tend to include too many irrelevant variables in Lasso-type procedures \citep{Mein:Bhlm:high:2006,Leng:Lin:Wahb:note:2006}, which was also observed in our simulations. Thus, it is natural to wonder whether the concave penalties really improve the variable selection performance over the Lasso and Enet, if the most appropriate amount of regularization for each method can be selected. In other words, it is sensible to compare the performance of the best model that ever exists on the solution path, assuming that we were assisted by an oracle which knew the true sparse model. Since our goal is to identify a parsimonious model that includes as many relevant variables as possible, we recorded the maximum number of correctly selected variables among all models up to a certain size on the solution path and averaged it over all replicates. By definition, this measure is an increasing function of the (maximum) model size and characterizes the best possible performance that can be achieved by any model selection criterion.

The results on the above defined performance measure for the first simulation study are shown in Figure~\ref{fig:low}. It is clear from the figure that the SICA outperformed the other methods in that it had a high chance to identify all six important variables immediately after the model size reached the true sparsity dimension. The MCP also exhibited a performance boost above the true sparsity level, while the boost for SCAD was more subtle and occurred at a later stage. These trends are magnified when the correlation among covariates increases. The Enet improved on the performance of the Lasso only marginally in this case.

\begin{center}
[Figure \ref{fig:low} about here.]
\end{center}

In the second simulation study, we compare the performance of various methods in the setting where $p$ is much larger than $n$. We used the same model setup as before, except that $p=2000$ and 5000, and $n=500$. The performance measures over 50 replicates with tenfold cross-validation to choose the optimal $\lambda$ are shown in Table \ref{tab:high}, which indicates the same trends as in Table~\ref{tab:low}. Note that in the most difficult setting, $p=5000$ and $\rho=0.5$, both the Lasso and Enet missed on average at least one important variable, whereas the concave penalties still included all six important variables. The variable selection performance assisted by an oracle is then compared in Figure~\ref{fig:high}. The SICA was still the winner in this case, followed by MCP. It is interesting to note that the Enet clearly outperformed the Lasso with moderate model size; this difference, however, was not notable in the cross-validated results. Since cross-validation seems to have less tolerance on false negatives than on false positives in these scenarios, its behavior is better explained by the tail of the curve: the SCAD had a slight performance boost on the tail, whereas the Lasso and Enet were very close to each other and both below the SCAD.

\begin{center}
[Table \ref{tab:high} about here.]
\end{center}

\begin{center}
[Figure \ref{fig:high} about here.]
\end{center}

The third simulation study reflects the more challenging case where the true model is only approximately sparse and we wish to analyze the robustness of our methods against departure from the sparsity assumption. To this end, in the previous model with $p=2000$ and $n=500$, we randomly chose 44 or 94 of the zero coefficients and perturbed them by $U(0,\ve)$ with signs randomly placed, so that $s=50$ or 100, respectively. This setting is intended to mimic a possible scenario in genetic studies where a relatively large number of genes have nonzero but weak effects, while a few key genes with strong effects stand out. To describe different levels of weak effects, we set $\ve=0.1$ and $0.2$. The dimensionality apparently exceeds the limit for exact recovery of all, strong and weak, signals; recall from condition \eqref{eq:wop_dim1} that we require $n\gg s^2\log p$ with constant $\varphi$ and $r_1=1$. In this case, even the oracle estimator, which is based on all covariates with nonzero effects, does not perform well. Thus, our target here is to identify the few key variables and a small portion of the variables with weak effects.

The results with tenfold cross-validation over 50 replicates are shown in Table~\ref{tab:hard}, where the added performance measure \#FN-S is the number of incorrectly excluded variables with strong effects. These results indicate that the SICA had the best performance in identifying the strong effects, because it aggressively seeks a sparse representation of the data and treats the variables with weak effects as noise variables. If the weak effects are also of interest, however, the SCAD and MCP performed better in that they selected more variables with nonzero effects at the expense of picking a larger model. It is worthwhile to note that when $\ve=0.1$, the SCAD, MCP, and SICA all had better prediction and estimation accuracy than the oracle estimator, which demonstrates the advantages of sparse modeling. The variable selection performance of all methods assisted by an oracle is depicted in Figure~\ref{fig:hard}, from which we see that the benefits of the concave penalties remain intact. Moreover, as the perturbation threshold $\ve$ increases, i.e., the model deviates further from sparsity, the SCAD and MCP showed less aggressive behavior than the SICA and were able to identify more variables with weak effects.

\begin{center}
[Table \ref{tab:hard} about here.]
\end{center}

\begin{center}
[Figure \ref{fig:hard} about here.]
\end{center}

Finally, we point out that there are also scenarios where the Enet tends to outperform the other methods; for instance, when many highly correlated variables have comparable effects and the sparsity dimension is intrinsically high, the $L_2$ term in the Enet penalty will play a pivotal role in regularization. These situations, however, are not our focus in sparse modeling and we do not pursue further in this article.

\section{Real Data Analysis}\label{sec:data}
We illustrate the proposed methods by an application to the diffuse large-B-cell lymphoma (DLBCL) data analyzed by \citet{Rose:Wrig:Chan:use:2002}. This data set consists of gene expression measurements for 7399 genes and survival outcomes after chemotherapy on 240 patients. The median follow-up time was 2.8 years and 138 patients died during the study period. The aim of the study was to formulate a molecular predictor of survival after chemotherapy for the disease. As in \citet{Rose:Wrig:Chan:use:2002}, the data set was randomly divided into a training set of 160 patients and a test set of 80 patients. We then applied the Lasso, SCAD, MCP, SICA, and Enet methods to the training set and used tenfold cross-validation to choose the optimal regularization parameter.

To assess the prediction performance of the resulting model, in addition to computing the prediction error based on the loss function $L(\cdot)$, we classified the test patients into a low-risk group and a high-risk group of equal size, according to the individual predicted excess risk $\widehat{\bbeta}^T\bZ$, and performed a two-sample log-rank test. Table \ref{tab:dlbcl_rslt} reports the number of selected genes, prediction error, and $p$-value of the log-rank test for each method, and the estimated coefficients for selected genes are given in Table \ref{tab:dlbcl_est}. These results suggest that all methods performed reasonably well in prediction, but the concave penalties selected sparser models than the Lasso and Enet. To see the similarity of the prediction results, we note that 60 of the 80 test patients received the same risk classification from all five methods. The selected genes also exhibit some consistency across different methods, although variability was observed in the cross-validation procedure due to the small sample size and ultra-high dimensionality. We remark that, to gain further stability, the idea of the sure independence screening \citep{Fan:Lv:sure:2008} can be used to reduce the dimensionality to a more manageable scale before applying the regularization methods, in which case the prediction performance of all methods tends to improve.

\begin{center}
[Table \ref{tab:dlbcl_rslt} about here.]
\end{center}

\begin{center}
[Table \ref{tab:dlbcl_est} about here.]
\end{center}

\section{Discussion}\label{sec:disc}
We have proposed a class of regularization methods for simultaneous variable selection and estimation in the additive hazards model. The main message of this article is that regularization methods can be used for variable selection with censored survival data in high-dimensional settings under conditions that are parallel to those in the linear regression context. Moreover, we have shown that concave penalties can substantially improve on the $L_1$ method and yield sparser models with better prediction performance, as evident from our theoretical and numerical results. Although we have focused on the additive hazards model, the empirical process techniques used here are fairly general and can be easily adapted to other survival models; for example, a key step to establishing a high-dimensional theory for the Cox model under similar conditions is to characterize the concentration of the empirical information matrix around its population counterpart.

The fact that our theoretical results allow the dimensionality to grow exponentially with the sample size has important implications. In practice, how high dimensionality the proposed methods can handle depends critically on the sample size and the correlation structure of the covariates. Variable selection in regression, in general, is a very difficult problem, and is even more challenging in the survival context. As a consequence, a relatively large sample size is essential to making reliable inference. In situations where the proposed methods may fail, it would be desirable to explore strategies that combine the strengths of a variety of approaches, and regularization methods can then be used as building blocks in such more powerful procedures.

We have partly based our performance comparisons on the best model selected by an oracle. In practice, how to choose the optimal regularization parameter remains a challenging issue for the regularization methodology. Although cross-validation is often the choice we have to resort to and sometimes performs better than AIC or BIC-type criteria, it suffers from several drawbacks and limitations. When the dimensionality is exceedingly high and the amount of regularization is necessarily large, the cross-validation curve can easily blow up and result in selecting too few variables. If variable selection is the sole purpose, stability selection \citep{Mein:Buhl:stab:2010} and related methods could potentially provide more accurate error control. These problems will be interesting topics for future research.

\appendix
\titleformat{\section}[block]{\large\bfseries\onehalfspacing}{\appendixname~\thesection: }{0pt}{}
\numberwithin{equation}{section}
\numberwithin{lem}{section}

\section{Proofs}
\setcounter{equation}{0}
\setcounter{lem}{0}

For clarity and readability, before proceeding to the proofs of our main results, we first present a lemma that provides optimality conditions for the regularization problem \eqref{eq:opt}, and establish a series of concentration inequalities that are essential to the main proofs. These results are also of independent interest. For notational simplicity, constants that are used in our proofs may vary from line to line.

\subsection{Optimality Conditions}
The following lemma provides optimality conditions that characterize a local solution to the regularization problem \eqref{eq:opt} and will be needed in the proof of Theorem \ref{thm:weak_oracle}.
\begin{lem}[Characterization of the regularized estimator]\label{lem:kkt}
Under Condition \ref{con:pen}, $\widehat{\bbeta}\in\mathbb{R}^p$ is a strict local minimizer of problem \eqref{eq:opt} if the following conditions hold:
\begin{gather}
\bU_{\widehat{A}}(\widehat{\bbeta})-\lambda\rho_{\lambda}'(|\widehat{\bbeta}_{\widehat{A}}|)\circ\sgn(\widehat{\bbeta}_{\widehat{A}})=\bzero,\label{eq:kkt1}\\
\|\bU_{\widehat{A}^c}(\widehat{\bbeta})\|_{\infty}<\lambda\rho'(0+),\label{eq:kkt2}
\end{gather}
and
\begin{equation}\label{eq:kkt3}
\Lambda_{\min}(\bV_{\widehat{A}\widehat{A}})>\lambda\kappa(\rho_{\lambda};\widehat{\bbeta}_{\widehat{A}}),
\end{equation}
where $\circ$ is the Hadamard (entrywise) product and the functions $|\cdot|$, $\rho_{\lambda}'(\cdot)$, and $\sgn(\cdot)$ are applied componentwise.
\end{lem}

\begin{proof}
We first consider the $|\widehat{A}|$-dimensional subspace $\mathcal{B}=\{\bbeta\in\mathbb{R}^p\colon\bbeta_{\widehat{A}^c}=\bzero\}$. Condition \eqref{eq:kkt3} implies that the objective function $Q(\bbeta)$ in \eqref{eq:opt} is strictly convex in a neighborhood of $\widehat{\bbeta}$ in $\mathcal{B}$. Then condition \eqref{eq:kkt1} implies that $\widehat{\bbeta}$ is a stationary point and hence a strict local minimizer of $Q(\bbeta)$ in the subspace $\mathcal{B}$.

It remains to show that, for any $\bbeta_1\in\mathbb{R}^p\setminus\mathcal{B}$ that lies in a sufficiently small neighborhood of $\widehat{\bbeta}$, we have $Q(\bbeta_1)>Q(\widehat{\bbeta})$. To this end, let $\bbeta_2$ be the projection of $\bbeta_1$ onto the subspace $\mathcal{B}$. Since $Q(\bbeta_2)\ge Q(\widehat{\bbeta})$ from the preceding paragraph, it suffices to show that $Q(\bbeta_1)>Q(\bbeta_2)$. By the mean value theorem, we have
\begin{equation}\label{eq:Q_diff}
Q(\bbeta_1)-Q(\bbeta_2)=\sum_{j\in\widehat{A}^c\colon\beta_{1j}\ne0}\frac{\partial Q(\bbeta^*)}{\partial\beta_j}\beta_{1j}
=\sum_{j\in\widehat{A}^c\colon\beta_{1j}\ne0}\{-U_j(\bbeta^*)+\lambda\rho_{\lambda}'(|\beta_j^*|)\sgn(\beta_j^*)\}\beta_{1j},
\end{equation}
where $\bbeta^*=(\beta_1^*,\cdots,\beta_p^*)^T$ is a point on the line segment between $\bbeta_1=(\beta_{11},\cdots,\beta_{1p})^T$ and $\bbeta_2$. It follows from condition \eqref{eq:kkt2} and continuity that $|U_j(\bbeta^*)|<\lambda\rho_{\lambda}'(|\beta_j^*|)\sgn(\beta_j^*)$ for all $j\in\widehat{A}^c$, provided that $\bbeta_1$, and hence $\bbeta^*$, is sufficiently close to $\widehat{\bbeta}$. Using this fact and that $\sgn(\beta_j^*)=\sgn(\beta_{1j})$, we see that each term in \eqref{eq:Q_diff} is positive, and thus $Q(\bbeta_1)>Q(\bbeta_2)$. This completes the proof.
\end{proof}

\subsection{Concentration Inequalities}\label{sec:pf_conc}
The major complexity in our proofs lies in characterizing the concentration of the large matrices $\bV$ and $\bW$ and vector $\bU(\bbeta_0)$. Due to the high dimensionality and dependency among entries, this is not a direct consequence of classical random matrix theory. Also, since each entry of the stochastic matrices or vector is not a sum of independent terms, the usual Hoeffding's inequality is not applicable. Hence, to establish the needed concentration inequalities, we rely on a functional Hoeffding-type inequality and some maximal inequalities for empirical processes as our primary mathematical tools.

We adopt the standard empirical process notation. For any measurable function $f$, we denote by $\mathbb{P}_nf$ and $Pf$ the expectations of $f$ under the empirical measure $\mathbb{P}_n$ and the probability measure $P$, respectively. Let $\|\cdot\|_{P,r}$ denote the usual $L_r(P)$-norm. The ``size'' of a class $\mathcal{F}$ of functions is measured by the \emph{bracketing number} $N_{[\,]}(\ve,\mathcal{F},L_r(P))$, the minimum number of $\ve$-brackets in $L_r(P)$ needed to cover $\mathcal{F}$, and the \emph{covering number} $N(\ve,\mathcal{F},L_2(Q))$, the minimum number of $L_2(Q)$-balls of radius $\ve$ needed to cover $\mathcal{F}$. The logarithms of the bracketing number and covering number are called \emph{entropy with bracketing} and \emph{entropy}, respectively. The \emph{bracketing integral} and \emph{uniform entropy integral} are defined as
\[
J_{[\,]}(\delta,\mathcal{F},L_2(P))=\int_0^{\delta}\sqrt{\log N_{[\,]}(\ve,\mathcal{F},L_2(P))}\,d\ve
\]
and
\[
J(\delta,\mathcal{F},L_2)=\int_0^{\delta}\sqrt{\log\sup_Q N(\ve\|F\|_{Q,2},\mathcal{F},L_2(Q))}\,d\ve,
\]
respectively, where $F$ is an envelope of $\mathcal{F}$, i.e., $|f|\le F$ for all $f\in\mathcal{F}$, and the supremum is taken over all probability measures $Q$ with $\|F\|_{Q,r}>0$. We refer the unfamiliar reader to Chapter 19 of \citet{Vaar:asym:1998} or Chapter 2 of \citet{Koso:intr:2008} for more definitions and concepts.

Define $\bS^{(k)}(t)=n^{-1}\sum_{j=1}^nY_j(t)\bZ_j(t)^{\otimes k}$, $k=0,1,2$, which are the sample counterparts of $\bs^{(k)}(t)$ defined in \eqref{eq:s_k}, and recall that $\overline{\bZ}(t)=\bS^{(1)}(t)/\bS^{(0)}(t)$. We begin with the following lemma, on which the other inequalities will be based.

\begin{lem}[Concentration of $\bS^{(k)}(\cdot)$, $k=0,1,2$]\label{lem:concS}
Under Condition \ref{con:proc}, there exist constants $C,K>0$ such that
\begin{align}
P\biggl(\sup_{t\in[0,\tau]}|S^{(0)}(t)-s^{(0)}(t)|\ge Cn^{-1/2}(1+x)\biggr)&\le\exp(-Kx^2),\label{eq:conc1}\\
P\biggl(\sup_{t\in[0,\tau]}|S_j^{(1)}(t)-s_j^{(1)}(t)|\ge Cn^{-1/2}(1+x)\mid\Omega_L\biggr)&\le\exp(-Kx^2/L^2),\label{eq:conc2}
\end{align}
and
\begin{equation}
P\biggl(\sup_{t\in[0,\tau]}|S_{ij}^{(2)}(t)-s_{ij}^{(2)}(t)|\ge Cn^{-1/2}(1+x)\mid\Omega_L\biggr)\le\exp(-Kx^2/L^4),\label{eq:conc3}
\end{equation}
for all $x>0$ and $i,j=1,\cdots,p$, where $S_j^{(1)}(\cdot)$ is the $j$th component of $\bS^{(1)}(\cdot)$ and $S_{ij}^{(2)}(\cdot)$ is the $(i,j)$th entry of the matrix $\bS^{(2)}(\cdot)$.
\end{lem}

\begin{proof}
We only show \eqref{eq:conc2}, and the other two inequalities follow similarly. Write $R_j=\sup_{t\in[0,\tau]}|S_j^{(1)}(t)-s_j^{(1)}(t)|$. The main idea is to apply a functional Hoeffding-type inequality, Theorem 9 of \citet{Mass:abou:2000}. To this end, we need to control the term $ER_j$.

We first show that the class of functions $\{Y(t)Z_j(t)\colon t\in[0,\tau]\}$ has bounded uniform entropy integral. Since a function of bounded variation can be expressed as the difference of two increasing functions, it follows from Lemma 9.10 of \citet{Koso:intr:2008} that $\mathcal{Z}_j\equiv\{Z_j(t)\colon t\in[0,\tau]\}$ is a VC-hull class associated with a VC class of index 2. Then, by Corollary 2.6.12 of \citet{Vaar:Well:weak:1996}, the entropy of $\mathcal{Z}_j$ satisfies $\log N(\ve\|F\|_{Q,2},\mathcal{Z}_j,L_2(Q))\le K'(1/\ve)$ for some constant $K'>0$, and hence $\mathcal{Z}_j$ has the uniform entropy integral
\[
J(1,\mathcal{Z}_j,L_2)\le\int_0^1\sqrt{K'(1/\ve)}\,d\ve<\infty.
\]
Also, by Example 19.16 of \citet{Vaar:asym:1998}, $\mathcal{Y}\equiv\{Y(t)\colon t\in[0,\tau]\}$ is a VC class and hence has bounded uniform entropy integral. Thus, by Theorem 9.15 of \citet{Koso:intr:2008}, $\mathcal{Y}\mathcal{Z}_j$ has bounded uniform entropy integral.

Now an application of Lemma 19.38 of \citet{Vaar:asym:1998} gives
\[
ER_j\le C'n^{-1/2}J(1,\mathcal{Y}\mathcal{Z}_j,L_2)\|F\|_{P,2}\le Cn^{-1/2}
\]
for some constants $C',C>0$, where the envelope $F$ is taken as $\sup_{t\in[0,\tau]}Y(t)|Z_j(t)|$. It follows from Theorem 9 of \citet{Mass:abou:2000} that
\[
P(R_j\ge Cn^{-1/2}(1+x)\mid\Omega_L)\le P(R_j\ge ER_j+Cn^{-1/2}x\mid\Omega_L)\le\exp(-Kx^2/L^2)
\]
for some constant $K>0$, which concludes the proof.
\end{proof}

\begin{lem}[Concentration of $\bU(\bbeta_0)$]\label{lem:concU}
Under Conditions \ref{con:proc}, there exist constants $C,D,\linebreak[1]K>0$ such that
\[
P(|U_j(\bbeta_0)|\ge Cn^{-1/2}(1+x)\mid\Omega_L)\le D\exp\biggl(-K\frac{x^2\wedge n}{L^4}\biggr)
\]
for all $x>0$ and $j=1,\cdots,p$, where $U_j(\bbeta_0)$ is the $j$th component of $\bU(\bbeta_0)$.
\end{lem}

\begin{proof}
We first write
\[
U_j(\bbeta_0)=\mathbb{P}_n\int_0^{\tau}\{Z_j(t)-\bar{Z}_j(t)\}\,dM(t)=\mathbb{P}_n\int_0^{\tau}Z_j(t)\,dM(t)-\mathbb{P}_n\int_0^{\tau}\bar{Z}_j(t)\,dM(t)\equiv T_1-T_2,
\]
where $\bar{Z}_j(\cdot)$ is the $j$th component of $\overline{\bZ}(\cdot)$. Since term $T_1$ is an i.i.d.\ sum of mean-zero random variables, an application of Hoeffding's inequality \citep{Hoef:prob:1963} gives $P(|T_1|\ge n^{-1/2}x\mid\Omega_L)\le2\exp(-Kx^2/L^4)$ for some constant $K>0$.

We will apply Theorem 9 of \citet{Mass:abou:2000} to bound term $T_2$. Note that, from \eqref{eq:conc1} and \eqref{eq:conc2} in Lemma \ref{lem:concS}, we have
\[
P\biggl(\sup_{t\in[0,\tau]}|S^{(0)}(t)-s^{(0)}(t)|\ge\delta\biggr)\le\exp(-Kn)
\]
and
\[
P\biggl(\sup_{t\in[0,\tau]}|S_j^{(1)}(t)-s_j^{(1)}(t)|\ge\delta\mid\Omega_L\biggr)\le\exp(-Kn/L^2),
\]
for some constant $\delta>0$ and $j=1,\cdots,p$. Since these two tail probabilities are bounded by $\exp(-Kn/L^4)$ for $L\ge1$, it suffices to consider the case where $\sup_{t\in[0,\tau]}|S^{(0)}(t)-s^{(0)}(t)|\le\delta$ and $\sup_{t\in[0,\tau]}|S_j^{(1)}(t)-s_j^{(1)}(t)|\le\delta$ for $j=1,\cdots,p$. Write
\begin{equation}\label{eq:Z_bar}
\bar{Z}_j(t)-e_j(t)=\frac{1}{S^{(0)}(t)}\{S_j^{(1)}(t)-s_j^{(1)}(t)\}-\frac{s_j^{(1)}(t)}{S^{(0)}(t)s^{(0)}(t)}\{S^{(0)}(t)-s^{(0)}(t)\}.
\end{equation}
Since $S^{(0)}(\cdot)$ and $s^{(0)}(\cdot)$ are bounded away from zero by Condition \ref{con:proc}(ii), the above representation implies that $\sup_{t\in[0,\tau]}|\bar{Z}_j(t)-e_j(t)|\le\delta'$ for some constant $\delta'>0$. Let $\mathcal{F}_j$ denote the class of functions $f\colon[0,\tau]\to\mathbb{R}$ that are of uniformly bounded variation and satisfy $\sup_{t\in[0,\tau]}|f(t)-e_j(t)|\le\delta'$. Define the class of functions $\mathcal{G}_j=\{\int_0^{\tau}f(t)\,dM(t)\colon f\in\mathcal{F}_j\}$ and $G_j=\sup_{g\in\mathcal{G}_j}|(\mathbb{P}_n-P)g|=\sup_{g\in\mathcal{G}_j}|\mathbb{P}_ng|$. We need to control $EG_j$.

By constructing $\|\cdot\|_{\infty}$-balls centered at piecewise constant functions on a regular grid, one can show that the covering number of the class $\mathcal{F}_j$ satisfies $N(\ve,\mathcal{F}_j,\|\cdot\|_{\infty})\le(K/\ve)^{K'/\ve}$ for some constants $K,K'>0$. Note also that, for any $f_1,f_2\in\mathcal{F}_j$,
\[
\biggl|\int_0^{\tau}f_1(t)\,dM(t)-\int_0^{\tau}f_2(t)\,dM(t)\biggr|\le\sup_{s\in[0,\tau]}|f_1(s)-f_2(s)|\int_0^{\tau}|dM(t)|.
\]
By Theorem 2.7.11 of \citet{Vaar:Well:weak:1996}, the bracketing number of the class $\mathcal{G}_j$ satisfies $N_{[\,]}(2\ve\|F\|_{P,2},\mathcal{G}_j,L_2(P))\le N(\ve,\mathcal{F}_j,\|\cdot\|_{\infty})\le(K/\ve)^{K'/\ve}$, where $F=\int_0^{\tau}|dM(t)|$. Hence, $\mathcal{G}_j$ has bounded bracketing integral.

An application of Corollary 19.35 of \citet{Vaar:asym:1998} yields that
\[
EG_j\le C'n^{-1/2}J_{[\,]}(\|G\|_{P,2},\mathcal{G}_j,L_2(P))\le Cn^{-1/2}
\]
for some constants $C',C>0$, where $G$ is an envelope of $\mathcal{G}_j$. We then apply Theorem 9 of \citet{Mass:abou:2000} to obtain $P(|T_2|\ge Cn^{-1/2}(1+x)\mid\Omega_L)\le\exp(-Kx^2/L^4)$ for some constant $K>0$. Putting the bounds for $T_1$ and $T_2$ together leads to the desired inequality, thus completing the proof.
\end{proof}

\begin{lem}[Concentration of $\bV$]\label{lem:concV-D}
Under Condition \ref{con:proc}, there exist constants $C,D,K>0$ such that
\[
P(|V_{ij}-D_{ij}|\ge Cn^{-1/2}(1+x)\mid\Omega_L)\le D\exp\biggl(-K\frac{x^2\wedge n}{L^4}\biggr)
\]
for all $x>0$ and $i,j=1,\cdots,p$, where $V_{ij}$ and $D_{ij}$ are the $(i,j)$th entries of the matrices $\bV$ and $\bD$, respectively.
\end{lem}

\begin{proof}
We first write
\[
V_{ij}-D_{ij}=\int_0^{\tau}\{S_{ij}^{(2)}(t)-s_{ij}^{(2)}(t)\}\,dt
+\int_0^{\tau}\biggl\{\frac{S_i^{(1)}(t)S_j^{(1)}(t)}{S^{(0)}(t)}-\frac{s_i^{(1)}(t)s_j^{(1)}(t)}{s^{(0)}(t)}\biggr\}\,dt
\equiv T_1+T_2.
\]
Clearly, \eqref{eq:conc3} in Lemma \ref{lem:concS} implies that $P(|T_1|\ge Cn^{-1/2}(1+x)\mid\Omega_L)\le\exp(-Kx^2/L^4)$. To bound term $T_2$, write
\begin{align*}
\frac{S_i^{(1)}(t)S_j^{(1)}(t)}{S^{(0)}(t)}-\frac{s_i^{(1)}(t)s_j^{(1)}(t)}{s^{(0)}(t)}
&=\frac{S_j^{(1)}(t)}{S^{(0)}(t)}\{S_i^{(1)}(t)-s_i^{(1)}(t)\}+\frac{s_i^{(1)}(t)}{S^{(0)}(t)}\{S_j^{(1)}(t)-s_j^{(1)}(t)\}\\
&\rel{=}{}-\frac{s_i^{(1)}(t)s_j^{(1)}(t)}{S^{(0)}(t)s^{(0)}(t)}\{S^{(0)}(t)-s^{(0)}(t)\}.
\end{align*}
By the same arguments as in the proof of Lemma \ref{lem:concU}, it suffices to consider the case where $\sup_{t\in[0,\tau]}|S^{(0)}(t)-s^{(0)}(t)|\le\delta$ and $\sup_{t\in[0,\tau]}|S_j^{(1)}(t)-s_j^{(1)}(t)|\le\delta$ for some constant $\delta>0$ and $j=1,\cdots,p$. From the above representation and \eqref{eq:conc1} and \eqref{eq:conc2} in Lemma \ref{lem:concS}, it follows that $P(|T_2|\ge Cn^{-1/2}(1+x)\mid\Omega_L)\le3\exp(-Kx^2/L^2)$. Combining the bounds for $T_1$ and $T_2$ yields the desired inequality and concludes the proof.
\end{proof}

\begin{lem}[Concentration of $\bW$]\label{lem:concW-Sigma}
Under Condition \ref{con:proc}, there exist constants $C,D,K>0$ such that
\[
P(|W_{ij}-\Sigma_{ij}|\ge Cn^{-1/2}(1+x)\mid\Omega_L)\le D\exp\biggl(-K\frac{x^2\wedge n}{L^4}\biggr)
\]
for all $x>0$ and $i,j=1,\cdots,p$, where $W_{ij}$ and $\Sigma_{ij}$ are the $(i,j)$th entries of the matrices $\bW$ and $\bSigma$, respectively.
\end{lem}

\begin{proof}
We first write
\begin{align*}
W_{ij}-\Sigma_{ij}&=(\mathbb{P}_n-P)\int_0^{\tau}Z_i(t)Z_j(t)\,dN(t)\\
&\rel{=}{}-\biggl\{\mathbb{P}_n\int_0^{\tau}\bar{Z}_i(t)Z_j(t)\,dN(t)-P\int_0^{\tau}e_i(t)Z_j(t)\,dN(t)\biggr\}\\
&\rel{=}{}-\biggl\{\mathbb{P}_n\int_0^{\tau}Z_i(t)\bar{Z}_j(t)\,dN(t)-P\int_0^{\tau}Z_i(t)e_j(t)\,dN(t)\biggr\}\\
&\rel{=}{}+\biggl\{\mathbb{P}_n\int_0^{\tau}\bar{Z}_i(t)\bar{Z}_j(t)\,dN(t)-P\int_0^{\tau}e_i(t)e_j(t)\,dN(t)\biggr\}\\
&\equiv T_1-T_2-T_3+T_4.
\end{align*}
Since term $T_1$ is an i.i.d.\ sum, an application of Hoeffding's inequality gives $P(|T_1|\ge n^{-1/2}x\mid\Omega_L)\le2\exp(-Kx^2/L^4)$. To bound term $T_2$, write
\[
T_2=(\mathbb{P}_n-P)\int_0^{\tau}\bar{Z}_i(t)Z_j(t)\,dN(t)+P\int_0^{\tau}\{\bar{Z}_i(t)-e_i(t)\}Z_j(t)\,dN(t)
\equiv T_{21}+T_{22}.
\]
Term $T_{21}$ can be bounded similarly as term $T_2$ in the proof of Lemma \ref{lem:concU}. Also, note that
\[
|T_{22}|\le\sup_{s\in[0,\tau]}|\bar{Z}_i(s)-e_i(s)|P\int_0^{\tau}|Z_j(t)|\,dN(t).
\]
Then it follows from \eqref{eq:Z_bar} and Lemma \ref{lem:concS} that
\[
P(|T_{22}|\ge Cn^{-1/2}(1+x)\mid\Omega_L)\le D\exp\biggl(-K\frac{x^2\wedge n}{L^4}\biggr).
\]
We can bound terms $T_3$ and $T_4$ similarly and thus obtain the desired inequality. This completes the proof.
\end{proof}

\subsection{Proof of Lemma \ref{lem:sample1}}
By the union bound and Lemma \ref{lem:concV-D}, we have
\begin{align*}
&P\biggl(\|\bV_{AA}-\bD_{AA}\|_{\infty}\ge\frac{1}{2\varphi}\Bigmid\Omega_L\biggr)\\
&\qquad=P\biggl(\max_{i\in A}\sum_{j\in A}|V_{ij}-D_{ij}|\ge\frac{1}{2\varphi}\Bigmid\Omega_L\biggr)
\le\sum_{i\in A}P\biggl(\sum_{j\in A}|V_{ij}-D_{ij}|\ge\frac{1}{2\varphi}\Bigmid\Omega_L\biggr)\\
&\qquad\le\sum_{i\in A}\sum_{j\in A}P\biggl(|V_{ij}-D_{ij}|\ge\frac{1}{2\varphi s}\Bigmid\Omega_L\biggr)
\le s^2D\exp\biggl\{-K\frac{n}{L^4}\biggl(\frac{1}{\varphi^2s^2}\wedge1\biggr)\biggr\}.
\end{align*}
By an error bound for matrix inversion \cite[p.~336]{Horn:John:matr:1985}, if $\|\bV_{AA}-\bD_{AA}\|_{\infty}<1/(2\varphi)$, then
\[
\frac{\|\bV_{AA}^{-1}-\bD_{AA}^{-1}\|_{\infty}}{\varphi}\le\frac{\varphi\|\bV_{AA}-\bD_{AA}\|_{\infty}}{1-\varphi\|\bV_{AA}-\bD_{AA}\|_{\infty}}<1,
\]
and hence $\|\bV_{AA}^{-1}\|_{\infty}\le\|\bD_{AA}^{-1}\|_{\infty}+\|\bV_{AA}^{-1}-\bD_{AA}^{-1}\|_{\infty}=\varphi+\|\bV_{AA}^{-1}-\bD_{AA}^{-1}\|_{\infty}<2\varphi$. Then probability bound \eqref{eq:phi} follows.

To show probability bound \eqref{eq:rho}, we write
\begin{align*}
\bV_{A^cA}\bV_{AA}^{-1}-\bD_{A^cA}\bD_{AA}^{-1}&=(\bV_{A^cA}-\bD_{A^cA})\bV_{AA}^{-1}+\bD_{A^cA}(\bV_{AA}^{-1}-\bD_{AA}^{-1})\\
&=(\bV_{A^cA}-\bD_{A^cA})\bV_{AA}^{-1}-\bD_{A^cA}\bD_{AA}^{-1}(\bV_{AA}-\bD_{AA})\bV_{AA}^{-1}\\
&\equiv T_1-T_2.
\end{align*}
Similarly as before, by the union bound and Lemma \ref{lem:concV-D}, we have
\begin{align*}
&P\biggl[\|\bV_{A^cA}-\bD_{A^cA}\|_{\infty}\ge\frac{1}{2\varphi}\biggl\{\frac{\alpha}{4}\frac{\rho'(0+)}{\rho_{\lambda}'(d)}\biggr\}\wedge \Bigl(\frac{c}{2}n^{\gamma}\Bigr)\Bigmid\Omega_L\biggr]\\
&\qquad\le(p-s)sD\exp\biggl[-K\frac{n}{L^4}\biggl\{\frac{\bigl(\rho_{\lambda}'(d)^{-1}\wedge n^{\gamma}\bigr)^2}{\varphi^2s^2}\wedge1\biggr\}\biggr].
\end{align*}
This, along with \eqref{eq:phi}, gives
\begin{align*}
&P\biggl[\|T_1\|_{\infty}\ge\biggl\{\frac{\alpha}{4}\frac{\rho'(0+)}{\rho_{\lambda}'(d)}\biggr\}\wedge\Bigl(\frac{c}{2}n^{\gamma}\Bigr)\Bigmid\Omega_L\biggr]\\
&\qquad\le P\biggl[\|\bV_{A^cA}-\bD_{A^cA}\|_{\infty}\ge\frac{1}{2\varphi}\biggl\{\frac{\alpha}{4}\frac{\rho'(0+)}{\rho_{\lambda}'(d)}\biggr\}\wedge \Bigl(\frac{c}{2}n^{\gamma}\Bigr)\Bigmid\Omega_L\biggr]+P\bigl(\|\bV_{AA}^{-1}\|_{\infty}\ge2\varphi\mid\Omega_L\bigr)\\
&\qquad\le(p-s)sD\exp\biggl[-K\frac{n}{L^4}\biggl\{\frac{\bigl(\rho_{\lambda}'(d)^{-1}\wedge n^{\gamma}\bigr)^2}{\varphi^2s^2}\wedge1\biggr\}\biggr]+s^2D\exp\biggl\{-K\frac{n}{L^4}\biggl(\frac{1}{\varphi^2s^2}\wedge1\biggr)\biggr\}.
\end{align*}
Also, by Condition \ref{con:inco} and \eqref{eq:phi}, we have
\begin{align*}
&P\biggl[\|T_2\|_{\infty}\ge\biggl\{\frac{\alpha}{4}\frac{\rho'(0+)}{\rho_{\lambda}'(d)}\biggr\}\wedge\Bigl(\frac{c}{2}n^{\gamma}\Bigr)\Bigmid\Omega_L\biggr]\\
&\qquad\le P\biggl[\|\bV_{AA}-\bD_{AA}\|_{\infty}\ge\frac{1}{2\varphi}\biggl\{\frac{\alpha}{4(1-\alpha)}\wedge\frac{1}{2}\biggr\}\Bigmid\Omega_L\biggr] +P\bigl(\|\bV_{AA}^{-1}\|_{\infty}\ge2\varphi\mid\Omega_L\bigr)\\
&\qquad\le s^2D\exp\biggl\{-K\frac{n}{L^4}\biggl(\frac{1}{\varphi^2s^2}\wedge1\biggr)\biggr\}.
\end{align*}
Putting the bounds for $T_1$ and $T_2$ together, we obtain
\begin{align*}
&P\biggl[\|\bV_{A^cA}\bV_{AA}^{-1}-\bD_{A^cA}\bD_{AA}^{-1}\|_{\infty}\ge\biggl\{\frac{\alpha}{2}\frac{\rho'(0+)}{\rho_{\lambda}'(d)}\biggr\}\wedge(cn^{\gamma}) \mid\Omega_L\biggr]\\
&\qquad\le(p-s)sD\exp\biggl[-K\frac{n}{L^4}\biggl\{\frac{\bigl(\rho_{\lambda}'(d)^{-1}\wedge n^{\gamma}\bigr)^2}{\varphi^2s^2}\wedge1\biggr\}\biggr]+s^2D\exp\biggl\{-K\frac{n}{L^4}\biggl(\frac{1}{\varphi^2s^2}\wedge1\biggr)\biggr\}.
\end{align*}
Then probability bound \eqref{eq:rho} follows from Condition \ref{con:inco} and the triangle inequality.

Finally, to show probability bound \eqref{eq:kappa}, by the Hoffman--Wielandt inequality \citep{Horn:John:matr:1985}, we have
\[
|\Lambda_{\min}(\bV_{AA})-\Lambda_{\min}(\bD_{AA})|\le\biggl\{\sum_{j=1}^s|\Lambda_{(j)}(\bV_{AA})-\Lambda_{(j)}(\bD_{AA})|^2\biggr\}^{1/2}\le\|\bV_{AA}-\bD_{AA}\|_F,
\]
where $\Lambda_{(j)}(\cdot)$ denotes the $j$th smallest eigenvalue and $\|\cdot\|_F$ is the Frobenius norm. Then it follows from the union bound and Lemma \ref{lem:concV-D} that
\begin{align*}
&P(|\Lambda_{\min}(\bV_{AA})-\Lambda_{\min}(\bD_{AA})|\ge\mu\mid\Omega_L)\\
&\qquad\le P(\|\bV_{AA}-\bD_{AA}\|_F\ge\mu\mid\Omega_L)=P\biggl(\sum_{i,j\in A}|V_{ij}-D_{ij}|^2\ge\mu^2\mid\Omega_L\biggr)\\
&\qquad\le\sum_{i,j\in A}P\biggl(|V_{ij}-D_{ij}|\ge\frac{\mu}{s}\Bigmid\Omega_L\biggr)\le s^2D\exp\biggl\{-K\frac{n}{L^4}\biggl(\frac{\mu^2}{s^2}\wedge1\biggr)\biggr\},
\end{align*}
which, by the definition of $\mu$, implies \eqref{eq:kappa}. This concludes the proof.

\subsection{Proof of Theorem \ref{thm:weak_oracle}}\label{sec:pf_weak}
The main idea of the proof is to first define an event with a high probability and then analyze the behavior of the regularized estimator $\widehat{\bbeta}$ conditional on that event. The first part involves concentration inequalities developed in Section \ref{sec:pf_conc} and Lemma \ref{lem:sample1}, while the second part involves nonprobabilistic arguments based on Lemma \ref{lem:kkt}.

First, by the union bound and Lemma \ref{lem:concU}, we have
\begin{multline}\label{eq:U_A}
P\biggl(\|\bU_A(\bbeta_0)\|_{\infty}\ge\frac{1}{2cn^{\gamma}}\frac{\alpha}{4}\lambda\rho'(0+)\mid\Omega_L\biggr)\\
\qquad\le\sum_{j\in A}P\biggl(|U_j(\bbeta_0)|\ge\frac{1}{2cn^{\gamma}}\frac{\alpha}{4}\lambda\rho'(0+)\mid\Omega_L\biggr)
\le sD\exp\biggl\{-K\frac{n}{L^4}\biggl(\frac{\lambda^2}{n^{2\gamma}}\wedge1\biggr)\biggr\}.
\end{multline}
Similarly, we have
\begin{equation}\label{eq:U_Ac}
P\biggl(\|\bU_{A^c}(\bbeta_0)\|_{\infty}\ge\frac{\alpha}{4}\lambda\rho'(0+)\mid\Omega_L\biggr)
\le(p-s)D\exp\biggl\{-K\frac{n}{L^4}(\lambda^2\wedge1)\biggr\}.
\end{equation}
Also, Condition \ref{con:proc}(iii) and the union bound imply that
\begin{equation}\label{eq:Omega}
P(\Omega_L^c)\le\sum_{j=1}^pP\biggl(\sup_{t\in[0,\tau]}|Z_j(t)|>L\biggr)\le pD\exp(-KL^r).
\end{equation}
It follows from \eqref{eq:U_A}--\eqref{eq:Omega} and Lemma \ref{lem:sample1} that with probability at least
\begin{multline}\label{eq:prob}
1-(p-s)sD\exp\biggl[-K\frac{n}{L^4}\biggl\{\frac{\bigl(\rho_{\lambda}'(d)^{-1}\wedge n^{\gamma}\bigr)^2}{\varphi^2s^2}\wedge1\biggr\}\biggr]\\
-s^2D\exp\biggl[-K\frac{n}{L^4}\biggl\{\frac{(\varphi^{-1}\wedge\mu)^2}{s^2}\wedge1\biggr\}\biggr] -sD\exp\biggl\{-K\frac{n}{L^4}\biggl(\frac{\lambda^2}{n^{2\gamma}}\wedge1\biggr)\biggr\}\\
-(p-s)D\exp\biggl\{-K\frac{n}{L^4}(\lambda^2\wedge1)\biggr\}-pD\exp(-KL^r),
\end{multline}
the following inequalities hold:
\begin{gather*}
\|\bU_A(\bbeta_0)\|_{\infty}<\frac{1}{2cn^{\gamma}}\frac{\alpha}{4}\lambda\rho'(0+),\qquad
\|\bU_{A^c}(\bbeta_0)\|_{\infty}<\frac{\alpha}{4}\lambda\rho'(0+),\\
\|\bV_{AA}^{-1}\|_{\infty}<2\varphi,\qquad
\|\bV_{A^cA}\bV_{AA}^{-1}\|_{\infty}<\biggl\{\Bigl(1-\frac{\alpha}{2}\Bigr)\frac{\rho'(0+)}{\rho_{\lambda}'(d)}\biggr\}\wedge(2cn^{\gamma}),
\end{gather*}
and
\begin{equation}\label{eq:eigen}
\Lambda_{\min}(\bV_{AA})>\lambda\kappa_0.
\end{equation}

From now on, we condition on the event that these inequalities hold. It suffices to find a $\widehat{\bbeta}\in\mathbb{R}^p$ that satisfies all the optimality conditions in Lemma \ref{lem:kkt} and the desired properties. To this end, take $\widehat{\bbeta}_{A^c}=\bzero$, and we will first determine $\widehat{\bbeta}_{A}$ by solving equation \eqref{eq:kkt1}. Since $\bU(\bbeta)=\bb-\bV\bbeta$, we have $\bU_A(\widehat{\bbeta})=\bU_A(\bbeta_0)-\bV_{AA}(\widehat{\bbeta}_A-\bbeta_{0A})$. Substituting this into the equation $\bU_A(\widehat{\bbeta})-\lambda\rho_{\lambda}'(|\widehat{\bbeta}_A|)\circ\sgn(\widehat{\bbeta}_A)=\bzero$ gives
\begin{equation}\label{eq:kkt1eq}
\widehat{\bbeta}_A-\bbeta_{0A}=\bV_{AA}^{-1}\{\bU_A(\bbeta_0)-\lambda\rho_{\lambda}'(|\widehat{\bbeta}_A|)\circ\sgn(\widehat{\bbeta}_A)\}.
\end{equation}
Define the function $f\colon\mathbb{R}^s\to\mathbb{R}^s$ by $f(\btheta)=\bbeta_{0A}+\bV_{AA}^{-1}\{\bU_A(\bbeta_0)-\lambda\rho_{\lambda}'(|\btheta|)\circ\sgn(\btheta)\}$, and let $\mathcal{K}$ denote the hypercube $\{\btheta\in\mathbb{R}^s\colon\|\btheta-\bbeta_{0A}\|_{\infty}\le c_1\varphi\lambda\rho'(0+)\}$. By the inequalities established before, for $\btheta\in\mathcal{K}$, we have
\begin{align*}
\|f(\btheta)-\bbeta_{0A}\|_{\infty}&\le\|\bV_{AA}^{-1}\|_{\infty}\{\|\bU_A(\bbeta_0)\|_{\infty}+\lambda\rho'(0+)\}\\
&\le2\varphi\biggl\{\frac{1}{2cn^{\gamma}}\frac{\alpha}{4}\lambda\rho'(0+)+\lambda\rho'(0+)\biggr\}
\le c_1\varphi\lambda\rho'(0+),
\end{align*}
i.e., $f(\mathcal{K})\subset\mathcal{K}$. Also, condition \eqref{eq:wop_d} implies that for $\btheta\in\mathcal{K}$, $\|\btheta-\bbeta_{0A}\|_{\infty}\le d$, so that $\sgn(\btheta)=\sgn(\bbeta_{0A})$. Thus, in view of Condition \ref{con:pen}, $f$ is a continuous function on the convex, compact hypercube $\mathcal{K}$. An application of Brouwer's fixed point theorem yields that equation \eqref{eq:kkt1eq} has a solution $\widehat{\bbeta}_A$ in $\mathcal{K}$. Moreover, $\sgn(\widehat{\bbeta}_A)=\sgn(\bbeta_{0A})$ and hence $\widehat{A}=A$. Thus, we have found a $\widehat{\bbeta}\in\mathbb{R}^p$ that satisfies \eqref{eq:kkt1} and the desired properties. It remains to check conditions \eqref{eq:kkt2} and \eqref{eq:kkt3}.

To verify that $\widehat{\bbeta}$ satisfies \eqref{eq:kkt2}, we write
\begin{align*}
\bU_{A^c}(\widehat{\bbeta})&=\bU_{A^c}(\bbeta_0)-\bV_{A^cA}(\widehat{\bbeta}_A-\bbeta_{0A})\\
&=\bU_{A^c}(\bbeta_0)-\bV_{A^cA}\bV_{AA}^{-1}\{\bU_A(\bbeta_0)-\lambda\rho_{\lambda}'(|\widehat{\bbeta}_A|)\circ\sgn(\widehat{\bbeta}_A)\},
\end{align*}
where we have substituted \eqref{eq:kkt1eq}. Since $\|\widehat{\bbeta}_A-\bbeta_{0A}\|_{\infty}\le d$, we have $\|\widehat{\bbeta}_A\|_{\infty}=\|\bbeta_{0A}+(\widehat{\bbeta}_A-\bbeta_{0A})\|_{\infty}\ge\|\bbeta_{0A}\|_{\infty}-\|\widehat{\bbeta}_A-\bbeta_{0A}\|_{\infty} \ge2d-d=d$. The triangle inequality, concavity of $\rho_{\lambda}(\cdot)$, and the inequalities established before together imply that
\begin{align*}
\|\bU_{A^c}(\widehat{\bbeta})\|_{\infty}&\le\|\bU_{A^c}(\bbeta_0)\|_{\infty}+\|\bV_{A^cA}\bV_{AA}^{-1}\|_{\infty}\{\|\bU_{A}(\bbeta_0)\|_{\infty} +\lambda\rho_{\lambda}'(d)\}\\
&<\frac{\alpha}{4}\lambda\rho'(0+)+2cn^{\gamma}\frac{1}{2cn^{\gamma}}\frac{\alpha}{4}\lambda\rho'(0+)+\Bigl(1-\frac{\alpha}{2}\Bigr)\frac{\rho'(0+)} {\rho_{\lambda}'(d)}\lambda\rho_{\lambda}'(d)\\
&=\frac{\alpha}{4}\lambda\rho'(0+)+\frac{\alpha}{4}\lambda\rho'(0+)+\Bigl(1-\frac{\alpha}{2}\Bigr)\lambda\rho'(0+)\\
&=\lambda\rho'(0+).
\end{align*}

Since $\|\widehat{\bbeta}_A-\bbeta_{0A}\|_{\infty}\le d$, it follows from \eqref{eq:eigen} that $\Lambda_{\min}(\bV_{AA})>\lambda\kappa_0\ge\lambda\kappa(\rho_{\lambda};\widehat{\bbeta}_A)$, and hence \eqref{eq:kkt3} is satisfied. Finally, we choose $L$ by matching the exponential terms in \eqref{eq:prob} and note that the probability tends to 1 by conditions \eqref{eq:wop_dim} and \eqref{eq:wop_lam}. This completes the proof.

\subsection{Proof of Theorem \ref{thm:oracle}}
The proof Theorem \ref{thm:oracle} is based on the proof of Theorem \ref{thm:weak_oracle} in Section \ref{sec:pf_weak}. First, by the same arguments as for deriving probability bound \eqref{eq:kappa} in Lemma \ref{lem:sample1}, one can obtain
\begin{align*}
&P(\Lambda_{\min}(\bV_{AA})\le\Lambda_1/2\mid\Omega_L)\\
&\qquad=P(|\Lambda_{\min}(\bV_{AA})-\Lambda_1|\ge\Lambda_1/2\mid\Omega_L)\le s^2D\exp\biggl\{-K\frac{n}{L^4}\biggl(\frac{\Lambda_1^2}{s^2}\wedge1\biggr)\biggr\}.
\end{align*}
Thus, with probability at least
\begin{equation}\label{eq:prob_or}
1-s^2D\exp\biggl\{-K\frac{n}{L^4}\biggl(\frac{\Lambda_1^2}{s^2}\wedge1\biggr)\biggr\}-pD\exp(-KL^r),
\end{equation}
it holds that $\Lambda_{\min}(\bV_{AA})>\Lambda_1/2$, and hence
\begin{equation}\label{eq:V_AA2}
\|\bV_{AA}^{-1}\|_2=1/\Lambda_{\min}(\bV_{AA})<2/\Lambda_1.
\end{equation}

From now on, we condition on the intersection of the event that \eqref{eq:V_AA2} holds and the event defined in the proof of Theorem \ref{thm:weak_oracle}; such an event still has a high probability. Since sparsity has been established in Theorem \ref{thm:weak_oracle}, we only need to show the asymptotic normality. By substituting \eqref{eq:kkt1eq}, we can write
\begin{align*}
\sqrt{n}\bu^T\bSigma_{AA}^{-1/2}\bD_{AA}(\widehat{\bbeta}_A-\bbeta_{0A})
&=\sqrt{n}\bu^T\bSigma_{AA}^{-1/2}\bD_{AA}\bV_{AA}^{-1}\{\bU_A(\bbeta_0)-\lambda\rho_{\lambda}'(|\widehat{\bbeta}_A|)\circ\sgn(\widehat{\bbeta}_A)\}\\
&=\sqrt{n}\bu^T\bSigma_{AA}^{-1/2}\bU_A(\bbeta_0)+\sqrt{n}\bu^T\bSigma_{AA}^{-1/2}\bD_{AA}(\bV_{AA}^{-1}-\bD_{AA}^{-1})\bU_A(\bbeta_0)\\
&\rel{=}{}-\sqrt{n}\bu^T\bSigma_{AA}^{-1/2}\bD_{AA}\bV_{AA}^{-1}\lambda\rho_{\lambda}'(|\widehat{\bbeta}_A|)\circ\sgn(\widehat{\bbeta}_A)\\
&\equiv T_1+T_2-T_3.
\end{align*}
We first consider term $T_2$. Since $\|\bu\|_2=1$, we have
\[
|T_2|\le\sqrt{n}\|\bSigma_{AA}^{-1/2}\bD_{AA}\|_2\|\bD_{AA}^{-1}\|_2\|\bV_{AA}-\bD_{AA}\|_2\|\bV_{AA}^{-1}\|_2\|\bU_A(\bbeta_0)\|_2.
\]
It follows from Lemma \ref{lem:concV-D} that
\begin{align*}
\|\bV_{AA}-\bD_{AA}\|_2\le\|\bV_{AA}-\bD_{AA}\|_F\le\biggl(s^2\max_{i,j\in A}|V_{ij}-D_{ij}|\biggr)^{1/2}=sO_p(n^{-1/2}),
\end{align*}
and similarly, by Lemma \ref{lem:concU}, $\|\bU_A(\bbeta_0)\|_2=\sqrt{s}O_p(n^{-1/2})$. Using also $\|\bSigma_{AA}^{-1/2}\bD_{AA}\|_2=\Lambda_3^{-1/2}$, $\|\bD_{AA}^{-1}\|_2=1/\Lambda_1$, and \eqref{eq:V_AA2}, we obtain
\[
|T_2|\le\sqrt{n}\Lambda_3^{-1/2}\Lambda_1^{-1}sO_p(n^{-1/2})2\Lambda_1^{-1}\sqrt{s}O_p(n^{-1/2})=\frac{2s^{3/2}}{\Lambda_1^2\Lambda_3^{1/2}}O_p(n^{-1/2})=o_p(1),
\]
by the third condition in \eqref{eq:op_dim}. We then consider term $T_3$. The concavity of $\rho_{\lambda}(\cdot)$, \eqref{eq:V_AA2}, and condition \eqref{eq:op_lam} lead to
\[
|T_3|\le\sqrt{n}\|\bSigma_{AA}^{-1/2}\bD_{AA}\|_2\|\bV_{AA}^{-1}\|_2\lambda\sqrt{s}\rho_{\lambda}'(d) <\frac{2\sqrt{ns}\lambda\rho_{\lambda}'(d)}{\Lambda_1\Lambda_3^{1/2}}\to0.
\]

It remains to show that term $T_1$ is asymptotically normal. Note that
\[
\bu^T\bSigma_{AA}^{-1/2}\bW_{AA}\bSigma_{AA}^{-1/2}\bu=1+\bu^T\bSigma_{AA}^{-1/2}(\bW_{AA}-\bSigma_{AA})\bSigma_{AA}^{-1/2}\bu.
\]
It follows from Lemma \ref{lem:concW-Sigma} that $\|\bW_{AA}-\bSigma_{AA}\|_2=sO_p(n^{-1/2})$. Then the second term in the above display is bounded by
\[
\|\bSigma_{AA}^{-1/2}\|_2\|\bW_{AA}-\bSigma_{AA}\|_2\|\bSigma_{AA}^{-1/2}\|_2=\Lambda_2^{-1/2}sO_p(n^{-1/2})\Lambda_2^{-1/2}=\frac{s}{\Lambda_2}O_p(n^{-1/2})=o_p(1),
\]
in view of the second condition in \eqref{eq:op_dim}. An application of the martingale central limit theorem \citep{Ande:Gill:cox':1982} yields that $T_1$ is asymptotically standard normal. Finally, we choose the optimal $L$ in \eqref{eq:prob_or}, note that the probability tends to 1 by the first condition in \eqref{eq:op_dim}, and combine it with the probability in Theorem \ref{thm:weak_oracle}. This concludes the proof.

\subsection{Proof of Theorem \ref{thm:conv}}
Since the sequence of objective functions $\{\widetilde{Q}(\bbeta^m)\}$ is decreasing, we see that the sequence $\{\widetilde{Q}(\bbeta^m)\}$ lies in a compact set of $\mathbb{R}$. This entails that $\{\bbeta^m\}$ is bounded, since any $|\beta_j|\to\infty$ would imply that $\widetilde{Q}(\bbeta)\to\infty$ by the assumption $V_{jj}\ne0$ for all $j$. Denote by $\balpha_j^m = (\beta_1^m, \cdots, \beta_{j-1}^m, \beta_j^m, \beta_{j+1}^{m-1}, \cdots, \beta_p^{m-1})^T$.

To show part (a), let $\bbeta^*$ be a cluster point of $\{\bbeta^m\}$ and $\{\bbeta^{m_k}\}$ a subsequence of $\{\bbeta^m\}$ that converges to $\bbeta^*$. We first prove a claim that if $\bbeta^{m_k}-\bbeta^{m_k-1}\to\bzero$, then $\bbeta^*$ is a local minimizer of $\widetilde{Q}(\bbeta)$. Denote by $\partial f(\bbeta;\balpha)$ the directional derivative of a function $f$ along the direction $\balpha$. For any $\btheta=(\theta_1,\cdots,\theta_p)\in\mathbb{R}^p$, we have
\begin{equation}\label{eq:dir_der}
\partial\widetilde{Q}(\bbeta^*;\btheta)=\sum_{j=1}^p\frac{\partial L(\bbeta^*)}{\partial\beta_j}\theta_j+\sum_{j=1}^pV_{jj}\partial p_{\lambda}(\beta_j^*;\theta_j)=\sum_{j=1}^p\partial\widetilde{Q}(\bbeta^*;\theta_j\be_j),
\end{equation}
where $\beta_j^*$ is the $j$th component of $\bbeta^*$ and $\be_j$ is the $p$-vector with 1 at the $j$th component and 0 elsewhere. Since $\beta_j^{m_k}$ minimizes $\widetilde{Q}(\beta_1^{m_k},\cdots,\beta_{j-1}^{m_k},\cdot,\beta_{j+1}^{m_k-1},\cdots,\beta_p^{m_k-1})$, we have $\partial\widetilde{Q}(\balpha_j^{m_k};\theta_j\be_j)\ge0$. Since $\bbeta^{m_k}-\bbeta^{m_k-1}\to\bzero$, we have $\lim_{k\to\infty}\bbeta^{m_k-1}=\lim_{k\to\infty}\bbeta^{m_k}=\bbeta^*$ and hence $\lim_{k\to\infty}\balpha_j^{m_k}=\bbeta^*$. It then follows from the upper semicontinuity of directional derivatives \citep{Bert:nonl:1999} that
\[
\partial\widetilde{Q}(\bbeta^*;\theta_j\be_j)\ge\limsup_{k\to\infty}\partial\widetilde{Q}(\balpha_j^{m_k};\theta_j\be_j)\ge0
\]
for all $j$. In view of \eqref{eq:dir_der}, we have $\partial\widetilde{Q}(\bbeta^*;\btheta)\ge0$ for all $\btheta\in\mathbb{R}^p$, so that $\bbeta^*$ is a local minimizer of $\widetilde{Q}(\bbeta)$. It remains to show that $\bbeta^{m_k}-\bbeta^{m_k-1}\to\bzero$.

In fact, we will show that $\bbeta^m-\bbeta^{m-1}\to\bzero$. Consider the update from $\balpha_{j-1}^m$ to $\balpha_j^m$. The condition $\kappa(p_{\lambda})<1$ implies that $\widetilde{Q}(\bbeta)$ is strictly convex in $\beta_j$, so that
\[
\widetilde{Q}(\balpha_{j-1}^m)-\widetilde{Q}(\balpha_j^m)\ge\theta(\beta_j^{m-1}-\beta_j^m)+\frac{c_0}{2}(\beta_j^{m-1}-\beta_j^m)^2,
\]
for every subgradient $\theta$ of $\widetilde{Q}(\beta_1^m,\cdots,\beta_{j-1}^m,\cdot,\beta_{j+1}^{m-1},\cdots, \beta_p^{m-1})$ at $\beta_j^m$, where $c_0=1-\kappa(p_{\lambda})>0$ \citep{Dem':Vasi:nond:1985}. The optimality of $\beta_j^m$ entails that 0 is a subgradient and hence
\[
\widetilde{Q}(\balpha_{j-1}^m)-\widetilde{Q}(\balpha_j^m)\ge\frac{c_0}{2}(\beta_j^{m-1}-\beta_j^m)^2.
\]
Adding both sides over $j=1,\cdots,p$ and $m=1,\cdots,M$ yields
\[
\widetilde{Q}(\bbeta^0)-\widetilde{Q}(\bbeta^M)\ge\frac{c_0}{2}\sum_{m=1}^M\|\bbeta^{m-1}-\bbeta^m\|_2^2.
\]
Noting that the sequence $\{\widetilde{Q}(\bbeta^m)\}$ is bounded, we have $\sum_{m=1}^{\infty}\|\bbeta^{m-1}-\bbeta^m\|_2^2 <\infty$, so that $\bbeta^m-\bbeta^{m-1}\to\bzero$. Part (a) is proved.

To show part (b), we use a contradiction argument. Suppose that the sequence $\{\bbeta^m\}$ does not converge to $\bbeta^*$. Then there exists a subsequence $\{\bbeta^{m_k}\}$ of $\{\bbeta^m\}$ such that $\|\bbeta^{m_k}-\bbeta^*\|_2 \ge \ve$ for some $\ve>0$. Since $\{\bbeta^{m_k}\}$ is bounded, by taking a further subsequence if necessary, we can assume that $\{\bbeta^{m_k}\}$ converges to a point $\bbeta^{**}$. Clearly $\bbeta^{**}\in\mathcal{K}$ since $\mathcal{K}$ is closed, and $\|\bbeta^{**}-\bbeta^*\|_2 \ge \ve$. It follows from part (a) that $\bbeta^{**}$ is a local minimizer of $\widetilde{Q}(\bbeta)$. This contradicts the assumption that $\bbeta^*$ is the unique local minimizer in $\mathcal{K}$ and proves part (b).

\subsection{Proof of Theorem \ref{thm:global}}
The inequality $\Lambda_{\min}(\bV_{SS})\ge\kappa(p_{\lambda})\max_{j\in S}V_{jj}$ in part (a) ensures that $\widetilde{Q}(\bbeta)$ is convex on the subspace $\mathcal{B}_S$, from which the restricted global optimality follows. Part (a) is proved.

To show part (b), note first that the strict inequality $\Lambda_{\min}(\bV_{SS})>\kappa(p_{\lambda})\max_{j\in S}V_{jj}$ implies strict convexity and hence the existence of a unique global minimizer $\bbeta^*$ of $\widetilde{Q}(\bbeta)$ on $\mathcal{B}_S$. Also, since $\Lambda_{\min}(\bV_{SS}) \le \min_{j\in S}V_{jj} \le \max_{j\in S}V_{jj}$, we have $\kappa(p_{\lambda})<1$. It then follows from Theorem \ref{thm:conv} and the boundedness of the sequence $\{\bbeta^m\}$ that the sequence $\{\bbeta^m\}$ converges to $\bbeta^*$. This proves part (b).

\section{SICA-Penalized Least Squares in One Dimension}\label{sec:app_sica}

In this appendix, we present an analytic form of the SICA-regularized estimator in one dimension. Consider the one-dimensional SICA-penalized least-squares problem
\[
\widehat{\theta}=\argmin_{\theta\in\mathbb{R}}\biggl\{\frac{1}{2}(\theta-\theta_0)^2+p_{\lambda}(|\theta|)\biggr\},
\]
where $\theta_0\in\mathbb{R}$ and $p_{\lambda}(\cdot)=\lambda\rho(\cdot)$ with the SICA penalty $\rho(\cdot)$ given by \eqref{eq:sica}. To find the nonzero stationary points of the objective function, it is easy to derive from the derivative equation that we need to solve the cubic equation $t^3+c_2t^2+c_1t+c_0=0$ for the positive roots, where $c_2=2a-|\theta_0|$, $c_1=a^2-2a|\theta_0|$, and $c_0=\lambda a(a+1)-a^2|\theta_0|$. Let $Q=(c_2^2-3c_1)/9$ and $R=(2c_2^3-9c_1c_2+27c_0)/54$. If $Q^3\le R^2$, the cubic equation either has a unique, negative root or, in addition, has a real root corresponding to an inflection point; hence $\widehat{\theta}=0$. Otherwise, compute the two largest roots
\[
t_1=-2\sqrt{Q}\cos\biggl(\frac{\alpha-2\pi}{3}\biggr)-\frac{c_2}{3}<t_2=-2\sqrt{Q}\cos\biggl(\frac{\alpha+2\pi}{3}\biggr)-\frac{c_2}{3},
\]
where $\alpha=\arccos(R/\sqrt{Q^3})$. If $t_1>0$, then $t_1$ is a local maximum and $t_2$ is a local minimum; by comparing the values of the objective function at 0 and $t_2$, we have that $\widehat{\theta}=\sgn(\theta_0)t_2$ if $t_2/2+\lambda(a+1)/(a+t_2)<\theta_0$, and $\widehat{\theta}=0$ otherwise. If $t_1<0$ and $t_2>0$, then $t_2$ is a local minimum and $\widehat{\theta}=\sgn(\theta_0)t_2$. Otherwise, the cubic equation has no positive roots and $\widehat{\theta}=0$.

In summary, $\widehat{\theta}=\sgn(\theta_0)t_2$ if $Q^3>R^2$, $t_1>0$, and $t_2/2+\lambda(a+1)/(a+t_2)<\theta_0$, or $Q^3>R^2$, $t_1<0$, and $t_2>0$; otherwise $\widehat{\theta}=0$. This gives a complete characterization of the SICA solution $\widehat{\theta}$ in one dimension.

\bibliographystyle{jasa}
\bibliography{stat}

\begin{thebibliography}{47}

\bibitem[{Andersen and Gill(1982)}]{Ande:Gill:cox':1982}
Andersen, P.~K. and Gill, R.~D. (1982), ``Cox's Regression Model for Counting
  Processes: A Large Sample Study,'' \textit{The Annals of Statistics}, 10,
  1100--1120.

\bibitem[{Antoniadis, Fryzlewicz,  and Letu\'{e}(2010)Antoniadis, Fryzlewicz,
  and Letu\'{e}}]{Anto:Fryz:Letu:2010}
Antoniadis, A., Fryzlewicz, P., and Letu\'{e}, F. (2010), ``The {D}antzig
  Selector in {C}ox's Proportional Hazards Model,'' \textit{Scandinavian
  Journal of Statistics}, 37, 531--552.

\bibitem[{Bertsekas(1999)}]{Bert:nonl:1999}
Bertsekas, D.~P. (1999), \textit{Nonlinear Programming} (2nd ed.), Belmont, MA:
  Athena Scientific.

\bibitem[{Bradic, Fan,  and Jiang(2011)Bradic, Fan, and
  Jiang}]{Brad:Fan:Jian:2011}
Bradic, J., Fan, J., and Jiang, J. (2011), ``Regularization for {C}ox's
  Proportional Hazards Model With {NP}-Dimensionality,'' \textit{The Annals of
  Statistics}, 39, 3092--3120.

\bibitem[{Breheny and Huang(2011)}]{Breh:Huan:coor:2011}
Breheny, P. and Huang, J. (2011), ``Coordinate Descent Algorithms for Nonconvex
  Penalized Regression, With Applications to Biological Feature Selection,''
  \textit{The Annals of Applied Statistics}, 5, 232--253.

\bibitem[{Breslow and Day(1987)}]{Bres:Day:stat:1987}
Breslow, N.~E. and Day, N.~E. (1987), \textit{Statistical Models in Cancer
  Research, 2: The Design and Analysis of Cohort Studies}, Lyon: IARC.

\bibitem[{Cai et~al.(2005)Cai, Fan, Li, and Zhou}]{Cai:Fan:Li:Zhou:vari:2005}
Cai, J., Fan, J., Li, R., and Zhou, H. (2005), ``Variable Selection for
  Multivariate Failure Time Data,'' \textit{Biometrika}, 92, 303--316.

\bibitem[{Cox and Oakes(1984)}]{Cox:Oake:anal:1984}
Cox, D.~R. and Oakes, D. (1984), \textit{Analysis of Survival Data}, London:
  Chapman \& Hall.

\bibitem[{Daubechies, Defrise,  and De~Mol(2004)Daubechies, Defrise, and
  De~Mol}]{Daub:Defr:De:2004}
Daubechies, I., Defrise, M., and De~Mol, C. (2004), ``An Iterative Thresholding
  Algorithm for Linear Inverse Problems With a Sparsity Constraint,''
  \textit{Communications on Pure and Applied Mathematics}, 57, 1413--1457.

\bibitem[{Dem'yanov and Vasil'ev(1985)}]{Dem':Vasi:nond:1985}
Dem'yanov, V.~F. and Vasil'ev, L.~V. (1985), \textit{Nondifferentiable
  Optimization}, New York: Springer.

\bibitem[{Fan(1997)}]{Fan:comm:1997}
Fan, J. (1997), ``Comments on `{W}avelets in Statistics: {A} Review,' by {A}.
  {A}ntoniadis,'' \textit{Journal of the Italian Statistical Society}, 6,
  131--138.

\bibitem[{Fan and Fan(2008)}]{Fan:Fan:fair:2008}
Fan, J. and Fan, Y. (2008), ``High-Dimensional Classification Using Features
  Annealed Independence Rules,'' \textit{The Annals of Statistics}, 36,
  2605--2637.

\bibitem[{Fan and Li(2001)}]{Fan:Li:vari:2001}
Fan, J. and Li, R. (2001), ``Variable Selection via Nonconcave Penalized
  Likelihood and Its Oracle Properties,'' \textit{Journal of the American
  Statistical Association}, 96, 1348--1360.

\bibitem[{Fan and Li(2002)}]{Fan:Li:vari:2002}
---------\quad (2002), ``Variable Selection for {C}ox's Proportional Hazards
  Model and Frailty Model,'' \textit{The Annals of Statistics}, 30, 74--99.

\bibitem[{Fan and Lv(2008)}]{Fan:Lv:sure:2008}
Fan, J. and Lv, J. (2008), ``Sure Independence Screening for Ultrahigh
  Dimensional Feature Space''  (with Discussion), \textit{Journal of the Royal
  Statistical Society, Ser.\ B}, 70, 849--911.

\bibitem[{Fan and Lv(2010)}]{Fan:Lv:sele:2010}
---------\quad (2010), ``A Selective Overview of Variable Selection in High
  Dimensional Feature Space (Invited Review Article),'' \textit{Statistica
  Sinica}, 20, 101--148.

\bibitem[{Fan and Lv(2011)}]{Fan:Lv:nonc:2011}
---------\quad (2011), ``Nonconcave Penalized Likelihood With
  {NP}-Dimensionality,'' \textit{IEEE Transactions on Information Theory}, 57,
  5467--5484.

\bibitem[{Fan, Lv,  and Qi(2011)Fan, Lv, and Qi}]{Fan:Lv:Qi:spar:2011}
Fan, J., Lv, J., and Qi, L. (2011), ``Sparse High-Dimensional Models in
  Economics,'' \textit{Annual Review of Economics}, 3, 291--317.

\bibitem[{Friedman et~al.(2007)Friedman, Hastie, H\"{o}fling, and
  Tibshirani}]{Frie:Hast:Hfli:Tibs:path:2007}
Friedman, J., Hastie, T., H\"{o}fling, H., and Tibshirani, R. (2007),
  ``Pathwise Coordinate Optimization,'' \textit{The Annals of Applied
  Statistics}, 1, 302--332.

\bibitem[{Fu(1998)}]{Fu:pena:1998}
Fu, W.~J. (1998), ``Penalized Regressions: {T}he Bridge Versus the Lasso,''
  \textit{Journal of Computational and Graphical Statistics}, 7, 397--416.

\bibitem[{Hoeffding(1963)}]{Hoef:prob:1963}
Hoeffding, W. (1963), ``Probability Inequalities for Sums of Bounded Random
  Variables,'' \textit{Journal of the American Statistical Association}, 58,
  13--30.

\bibitem[{Horn and Johnson(1985)}]{Horn:John:matr:1985}
Horn, R.~A. and Johnson, C.~R. (1985), \textit{Matrix Analysis}, New York:
  Cambridge Univ.\ Press.

\bibitem[{Jarrow(2009)}]{Jarrow:cred:2009}
Jarrow, R.~A. (2009), ``Credit Risk Models,'' \textit{Annual Review of
  Financial Economics}, 1, 37--68.

\bibitem[{Kosorok(2008)}]{Koso:intr:2008}
Kosorok, M.~R. (2008), \textit{Introduction to Empirical Processes and
  Semiparametric Inference}, New York: Springer.

\bibitem[{Leng, Lin,  and Wahba(2006)Leng, Lin, and
  Wahba}]{Leng:Lin:Wahb:note:2006}
Leng, C., Lin, Y., and Wahba, G. (2006), ``A Note on the {L}asso and Related
  Procedures in Model Selection,'' \textit{Statistica Sinica}, 16, 1273--1284.

\bibitem[{Leng and Ma(2007)}]{Leng:Ma:path:2007}
Leng, C. and Ma, S. (2007), ``Path Consistent Model Selection in Additive Risk
  Model via {L}asso,'' \textit{Statistics in Medicine}, 26, 3753--3770.

\bibitem[{Lin and Ying(1994)}]{Lin:Ying:semi:1994}
Lin, D.~Y. and Ying, Z. (1994), ``Semiparametric Analysis of the Additive Risk
  Model,'' \textit{Biometrika}, 81, 61--71.

\bibitem[{Lv and Fan(2009)}]{Lv:Fan:unif:2009}
Lv, J. and Fan, Y. (2009), ``A Unified Approach to Model Selection and Sparse
  Recovery Using Regularized Least Squares,'' \textit{The Annals of
  Statistics}, 37, 3498--3528.

\bibitem[{Martinussen and Scheike(2009)}]{Mart:Sche:cova:2009}
Martinussen, T. and Scheike, T.~H. (2009), ``Covariate Selection for the
  Semiparametric Additive Risk Model,'' \textit{Scandinavian Journal of
  Statistics}, 36, 602--619.

\bibitem[{Massart(2000)}]{Mass:abou:2000}
Massart, P. (2000), ``About the Constants in {T}alagrand's Concentration
  Inequalities for Empirical Processes,'' \textit{The Annals of Probability},
  28, 863--884.

\bibitem[{Mazumder, Friedman,  and Hastie(2011)Mazumder, Friedman, and
  Hastie}]{Mazu:Frie:Hast:spar:2011}
Mazumder, R., Friedman, J.~H., and Hastie, T. (2011), ``{\it SparseNet}:
  {C}oordinate Descent With Nonconvex Penalties,'' \textit{Journal of the
  American Statistical Association}, 106, 1125--1138.

\bibitem[{Meinshausen and B\"{u}hlmann(2006)}]{Mein:Bhlm:high:2006}
Meinshausen, N. and B\"{u}hlmann, P. (2006), ``High-Dimensional Graphs and
  Variable Selection With the {L}asso,'' \textit{The Annals of Statistics}, 34,
  1436--1462.

\bibitem[{Meinshausen and B\"{u}hlmann(2010)}]{Mein:Buhl:stab:2010}
---------\quad (2010), ``Stability Selection''  (with Discussion),
  \textit{Journal of the Royal Statistical Society, Ser.\ B}, 72, 417--473.

\bibitem[{Rosenwald et~al.(2002)Rosenwald, Wright, Chan, Connors, Campo,
  Fisher, Gascoyne, Muller-Hermelink, Smeland, and
  Staudt}]{Rose:Wrig:Chan:use:2002}
Rosenwald, A., Wright, G., Chan, W.~C., Connors, J.~M., Campo, E., Fisher,
  R.~I., Gascoyne, R.~D., Muller-Hermelink, H.~K., Smeland, E.~B., and Staudt,
  L.~M. (2002), ``The Use of Molecular Profiling to Predict Survival after
  Chemotherapy for Diffuse Large-{B}-Cell Lymphoma,'' \textit{The New England
  Journal of Medicine}, 346, 1937--1947.

\bibitem[{Tibshirani(1996)}]{Tibs:regr:1996}
Tibshirani, R. (1996), ``Regression Shrinkage and Selection via the Lasso,''
  \textit{Journal of the Royal Statistical Society, Ser.\ B}, 58, 267--288.

\bibitem[{Tibshirani(1997)}]{Tibs:lass:1997}
---------\quad (1997), ``The Lasso Method for Variable Selection in the {C}ox
  Model,'' \textit{Statistics in Medicine}, 16, 385--395.

\bibitem[{van~der Vaart(1998)}]{Vaar:asym:1998}
van~der Vaart, A.~W. (1998), \textit{Asymptotic Statistics}, New York:
  Cambridge Univ.\ Press.

\bibitem[{van~der Vaart and Wellner(1996)}]{Vaar:Well:weak:1996}
van~der Vaart, A.~W. and Wellner, J.~A. (1996), \textit{Weak Convergence and
  Empirical Processes: With Applications to Statistics}, New York: Springer.

\bibitem[{Wainwright(2009)}]{Wain:shar:2009}
Wainwright, M.~J. (2009), ``Sharp Thresholds for High-Dimensional and Noisy
  Sparsity Recovery Using $\ell_1$-Constrained Quadratic Programming
  ({L}asso),'' \textit{IEEE Transactions on Information Theory}, 55,
  2183--2202.

\bibitem[{Wu and Lange(2008)}]{Wu:Lang:coor:2008}
Wu, T.~T. and Lange, K. (2008), ``Coordinate Descent Algorithms for Lasso
  Penalized Regression,'' \textit{The Annals of Applied Statistics}, 2,
  224--244.

\bibitem[{Zhang(2010)}]{Zhan:near:2010}
Zhang, C.-H. (2010), ``Nearly Unbiased Variable Selection Under Minimax Concave
  Penalty,'' \textit{The Annals of Statistics}, 38, 894--942.

\bibitem[{Zhang and Lu(2007)}]{Zhan:Lu:adap:2007}
Zhang, H.~H. and Lu, W. (2007), ``Adaptive {L}asso for {C}ox's Proportional
  Hazards Model,'' \textit{Biometrika}, 94, 691--703.

\bibitem[{Zhao and Yu(2006)}]{Zhao:Yu:on:2006}
Zhao, P. and Yu, B. (2006), ``On Model Selection Consistency of {L}asso,''
  \textit{Journal of Machine Learning Research}, 7, 2541--2563.

\bibitem[{Zou(2006)}]{Zou:adap:2006}
Zou, H. (2006), ``The Adaptive Lasso and Its Oracle Properties,''
  \textit{Journal of the American Statistical Association}, 101, 1418--1429.

\bibitem[{Zou(2008)}]{Zou:note:2008}
---------\quad (2008), ``A Note on Path-Based Variable Selection in the
  Penalized Proportional Hazards Model,'' \textit{Biometrika}, 95, 241--247.

\bibitem[{Zou and Hastie(2005)}]{Zou:Hast:regu:2005}
Zou, H. and Hastie, T. (2005), ``Regularization and Variable Selection via the
  Elastic Net,'' \textit{Journal of the Royal Statistical Society, Ser.\ B},
  67, 301--320.

\bibitem[{Zou and Li(2008)}]{Zou:Li:one-:2008}
Zou, H. and Li, R. (2008), ``One-Step Sparse Estimates in Nonconcave Penalized
  Likelihood Models (with Discussion),'' \textit{The Annals of Statistics}, 36,
  1509--1566.

\end{thebibliography}

\begin{sidewaystable}
\def\~{\phantom{0}}\centering
\caption{Results for various methods in the first simulation study with $n=200$, $s=6$, and censoring rate about $25\%$. Values shown are means (standard deviations) of each performance measure over 100 replicates}\label{tab:low}\vskip1ex
\begin{tabularx}{\linewidth}{@{}XX>{\hspace{7pt}}c*{5}{>{\hspace{14pt}}c}@{}}
\midrule
Setting    & Method & PE1 & PE2 & $L_2$-loss & $L_1$-loss & \#S & \#FN\\
\midrule
$p=50$     & Lasso  & 0.191 (0.055) & 22.73 (6.38) & 1.061 (0.311) & 3.341 (0.838) & 19.6 (4.5) & 0.0 (0.2)\\
$\rho=0.1$ & SCAD   & 0.135 (0.044) & 10.90 (5.24) & 0.498 (0.250) & 1.266 (0.658) & 10.7 (2.5) & 0.0 (0.2)\\
           & MCP    & 0.138 (0.047) & 11.34 (5.83) & 0.518 (0.278) & 1.271 (0.716) &\~8.9 (2.1) & 0.0 (0.3)\\
           & SICA   & 0.130 (0.048) & 10.30 (5.77) & 0.471 (0.273) & 1.015 (0.665) &\~6.2 (1.0) & 0.1 (0.4)\\
           & Enet   & 0.192 (0.055) & 22.75 (6.32) & 1.062 (0.308) & 3.367 (0.896) & 19.9 (4.9) & 0.0 (0.1)\\
           & Oracle & 0.121 (0.030) &\~9.26 (3.62) & 0.424 (0.172) & 0.894 (0.398) &\~6.0 (0.0) & 0.0 (0.0)\\
\addlinespace
$p=50$     & Lasso  & 0.199 (0.052) & 18.33 (4.58) & 1.078 (0.310) & 3.504 (0.916) & 21.2 (5.3) & 0.1 (0.5)\\
$\rho=0.5$ & SCAD   & 0.134 (0.071) &\~9.98 (5.63) & 0.522 (0.354) & 1.370 (1.089) & 10.7 (3.0) & 0.0 (0.2)\\
           & MCP    & 0.130 (0.057) &\~9.92 (5.77) & 0.522 (0.374) & 1.295 (1.044) &\~8.6 (2.4) & 0.1 (0.7)\\
           & SICA   & 0.121 (0.040) &\~8.87 (4.17) & 0.461 (0.246) & 1.058 (0.788) &\~6.9 (2.8) & 0.0 (0.0)\\
           & Enet   & 0.199 (0.052) & 18.32 (4.44) & 1.078 (0.300) & 3.548 (1.049) & 21.8 (5.9) & 0.0 (0.0)\\
           & Oracle & 0.108 (0.019) &\~7.56 (2.96) & 0.398 (0.187) & 0.850 (0.452) &\~6.0 (0.0) & 0.0 (0.0)\\
\addlinespace
$p=100$    & Lasso  & 0.297 (0.060) & 25.44 (5.31) & 1.289 (0.284) & 4.312 (0.932) & 25.2 (6.2) & 0.0 (0.1)\\
$\rho=0.1$ & SCAD   & 0.253 (0.054) & 11.72 (5.01) & 0.558 (0.257) & 1.601 (0.740) & 15.1 (4.0) & 0.0 (0.2)\\
           & MCP    & 0.253 (0.059) & 11.83 (5.39) & 0.563 (0.275) & 1.507 (0.801) & 11.5 (3.1) & 0.0 (0.3)\\
           & SICA   & 0.236 (0.062) & 10.13 (5.44) & 0.483 (0.270) & 1.070 (0.751) &\~6.7 (2.5) & 0.0 (0.2)\\
           & Enet   & 0.300 (0.066) & 25.70 (5.79) & 1.302 (0.307) & 4.403 (1.137) & 26.0 (7.7) & 0.0 (0.1)\\
           & Oracle & 0.219 (0.041) &\~8.77 (3.90) & 0.423 (0.203) & 0.892 (0.472) &\~6.0 (0.0) & 0.0 (0.0)\\
\addlinespace
$p=100$    & Lasso  & 0.275 (0.057) & 22.98 (4.82) & 1.379 (0.334) & 4.716 (0.931) & 27.5 (7.0) & 0.2 (0.8)\\
$\rho=0.5$ & SCAD   & 0.207 (0.053) & 11.24 (5.87) & 0.589 (0.393) & 1.697 (1.055) & 15.3 (4.2) & 0.1 (0.8)\\
           & MCP    & 0.209 (0.054) & 11.48 (6.29) & 0.606 (0.420) & 1.643 (1.115) & 11.5 (3.1) & 0.2 (1.0)\\
           & SICA   & 0.198 (0.071) & 10.03 (6.75) & 0.542 (0.417) & 1.264 (1.151) &\~6.7 (2.6) & 0.2 (0.8)\\
           & Enet   & 0.275 (0.058) & 23.07 (4.96) & 1.385 (0.343) & 4.733 (0.974) & 27.6 (7.1) & 0.2 (0.8)\\
           & Oracle & 0.169 (0.023) &\~7.42 (2.71) & 0.394 (0.179) & 0.838 (0.427) &\~6.0 (0.0) & 0.0 (0.0)\\
\midrule
\end{tabularx}
\end{sidewaystable}

\begin{sidewaystable}
\def\~{\phantom{0}}\centering
\caption{Results for various methods in the second simulation study with $n=500$, $s=6$, and censoring rate about $25\%$. Values shown are means (standard deviations) of each performance measure over 50 replicates}\label{tab:high}\vskip1ex
\begin{tabularx}{\linewidth}{@{}XX>{\hspace{7pt}}c*{5}{>{\hspace{14pt}}c}@{}}
\midrule
Setting    & Method & PE1 & PE2 & $L_2$-loss & $L_1$-loss & \#S & \#FN\\
\midrule
$p=2000$   & Lasso  & 10.673 (0.038) & 27.43 (3.99) & 1.396 (0.208) & 5.036 (0.780) & 58.4 (14.8) & 0.0 (0.0)\\
$\rho=0.1$ & SCAD   & 10.558 (0.017) &\~6.48 (1.95) & 0.301 (0.096) & 0.972 (0.435) & 21.0 (9.3)\~& 0.0 (0.0)\\
           & MCP    & 10.556 (0.020) &\~6.29 (2.56) & 0.293 (0.123) & 0.772 (0.492) & 11.5 (4.9)\~& 0.0 (0.0)\\
           & SICA   & 10.551 (0.014) &\~5.64 (2.48) & 0.265 (0.120) & 0.582 (0.376) &  6.7 (3.6)  & 0.0 (0.0)\\
           & Enet   & 10.673 (0.038) & 27.43 (3.99) & 1.396 (0.208) & 5.036 (0.780) & 58.4 (14.8) & 0.0 (0.0)\\
           & Oracle & 10.548 (0.010) &\~5.16 (1.87) & 0.244 (0.097) & 0.511 (0.223) &  6.0 (0.0)  & 0.0 (0.0)\\
\addlinespace
$p=2000$   & Lasso  & 11.913 (0.051) & 27.24 (3.85) & 1.638 (0.250) & 5.809 (0.747) & 65.4 (18.2) & 0.1 (0.3)\\
$\rho=0.5$ & SCAD   & 11.743 (0.022) &\~7.30 (2.21) & 0.362 (0.117) & 1.409 (0.498) & 34.3 (12.8) & 0.0 (0.0)\\
           & MCP    & 11.738 (0.019) &\~6.35 (2.19) & 0.322 (0.119) & 0.956 (0.408) & 15.9 (5.9)\~& 0.0 (0.0)\\
           & SICA   & 11.732 (0.022) &\~5.19 (2.32) & 0.274 (0.123) & 0.609 (0.339) &  6.8 (4.0)  & 0.0 (0.0)\\
           & Enet   & 11.913 (0.051) & 27.24 (3.85) & 1.638 (0.250) & 5.809 (0.747) & 65.4 (18.2) & 0.1 (0.3)\\
           & Oracle & 11.728 (0.009) &\~4.92 (1.79) & 0.264 (0.108) & 0.558 (0.238) &  6.0 (0.0)  & 0.0 (0.0)\\
\addlinespace
$p=5000$   & Lasso  & 11.320 (0.047) & 30.62 (4.21) & 1.610 (0.223) & 5.559 (0.790) & 59.2 (18.2) & 0.0 (0.1)\\
$\rho=0.1$ & SCAD   & 11.138 (0.016) &\~8.23 (2.36) & 0.391 (0.122) & 1.611 (0.570) & 38.2 (14.5) & 0.0 (0.0)\\
           & MCP    & 11.132 (0.017) &\~7.23 (2.44) & 0.348 (0.126) & 1.029 (0.435) & 16.5 (7.1)\~& 0.0 (0.0)\\
           & SICA   & 11.123 (0.011) &\~6.06 (2.52) & 0.298 (0.133) & 0.648 (0.351) &  6.2 (1.0)  & 0.0 (0.0)\\
           & Enet   & 11.321 (0.048) & 30.62 (4.22) & 1.610 (0.223) & 5.571 (0.824) & 59.4 (18.6) & 0.0 (0.1)\\
           & Oracle & 11.122 (0.010) &\~5.78 (2.25) & 0.283 (0.121) & 0.608 (0.291) &  6.0 (0.0)  & 0.0 (0.0)\\
\addlinespace
$p=5000$   & Lasso  & 10.364 (0.053) & 31.81 (4.30) & 1.997 (0.293) & 6.390 (0.553) & 56.9 (28.7) & 1.1 (1.8)\\
$\rho=0.5$ & SCAD   & 10.120 (0.018) &\~9.24 (2.39) & 0.464 (0.151) & 2.424 (0.839) & 64.2 (20.1) & 0.0 (0.0)\\
           & MCP    & 10.111 (0.019) &\~7.69 (2.38) & 0.388 (0.136) & 1.471 (0.643) & 27.8 (10.8) & 0.0 (0.0)\\
           & SICA   & 10.096 (0.010) &\~4.91 (1.90) & 0.267 (0.115) & 0.573 (0.272) &  6.1 (0.2)  & 0.0 (0.0)\\
           & Enet   & 10.364 (0.053) & 31.81 (4.30) & 1.997 (0.293) & 6.387 (0.532) & 56.9 (28.0) & 1.0 (1.7)\\
           & Oracle & 10.094 (0.007) &\~4.65 (1.65) & 0.255 (0.110) & 0.552 (0.270) &  6.0 (0.0)  & 0.0 (0.0)\\
\midrule
\end{tabularx}
\end{sidewaystable}

\begin{sidewaystable}
\def\~{\phantom{0}}\centering
\caption{Results for various methods in the third simulation study with $n=500$, $p=2000$, and censoring rate about $25\%$. Values shown are means (standard deviations) of each performance measure over 50 replicates. The oracle method is based on all nonzero (strong and weak) effects}\label{tab:hard}\vskip1ex
\begin{tabularx}{\linewidth}{@{}XX>{\hspace{3pt}}c*{6}{>{\hspace{6pt}}c}@{}}
\midrule
Setting   & Method & PE1 & PE2 & $L_2$-loss & $L_1$-loss & \#S & \#FN & \#FN-S\\
\midrule
$s=50$    & Lasso  &\~9.743 (0.039) & 29.80 (3.28) & 1.808 (0.214) &\~8.461 (0.675) & 63.8 (17.5) & 42.5 (1.4) & 0.1 (0.6)\\
$\ve=0.1$ & SCAD   &\~9.597 (0.018) & 12.07 (1.26) & 0.552 (0.073) &\~4.039 (0.692) & 43.5 (14.5) & 41.6 (1.4) & 0.0 (0.0)\\
          & MCP    &\~9.591 (0.017) & 11.74 (1.25) & 0.546 (0.082) &\~3.568 (0.575) & 20.4 (8.6)\~& 43.0 (0.9) & 0.0 (0.0)\\
          & SICA   &\~9.582 (0.009) & 11.16 (1.04) & 0.536 (0.090) &\~3.168 (0.352) &  6.0 (0.0)  & 44.0 (0.0) & 0.0 (0.0)\\
          & Enet   &\~9.743 (0.040) & 29.80 (3.28) & 1.809 (0.215) &\~8.463 (0.676) & 63.9 (17.3) & 42.5 (1.4) & 0.1 (0.6)\\
          & Oracle &\~9.615 (0.046) & 14.59 (2.72) & 0.692 (0.145) &\~3.905 (0.794) & 50.0 (0.0)\~&\~0.0 (0.0) & 0.0 (0.0)\\
\addlinespace
$s=50$    & Lasso  & 10.517 (0.040) & 38.02 (2.95) & 2.288 (0.205) & 11.531 (0.749) & 47.4 (29.9) & 43.5 (3.4) & 1.4 (2.0)\\
$\ve=0.2$ & SCAD   & 10.359 (0.057) & 22.17 (4.90) & 1.097 (0.360) &\~7.986 (1.175) & 66.3 (17.7) & 38.0 (3.3) & 0.1 (0.8)\\
          & MCP    & 10.364 (0.071) & 23.33 (6.30) & 1.200 (0.460) &\~7.701 (1.494) & 33.4 (10.5) & 40.7 (2.9) & 0.3 (1.1)\\
          & SICA   & 10.327 (0.023) & 21.63 (1.28) & 1.121 (0.110) &\~6.755 (0.607) &  7.1 (6.1)  & 43.9 (0.4) & 0.0 (0.0)\\
          & Enet   & 10.518 (0.040) & 38.04 (2.94) & 2.289 (0.205) & 11.547 (0.769) & 48.4 (30.5) & 43.4 (3.4) & 1.3 (2.0)\\
          & Oracle & 10.252 (0.045) & 15.22 (2.46) & 0.717 (0.129) &\~4.005 (0.649) & 50.0 (0.0)\~&\~0.0 (0.0) & 0.0 (0.0)\\
\addlinespace
$s=100$   & Lasso  & 12.630 (0.045) & 33.34 (3.47) & 2.026 (0.229) & 10.942 (0.689) & 57.0 (23.1) & 91.6 (2.7) & 0.6 (1.3)\\
$\ve=0.1$ & SCAD   & 12.460 (0.023) & 14.87 (1.17) & 0.712 (0.073) &\~6.706 (0.641) & 53.7 (14.8) & 89.1 (2.5) & 0.0 (0.0)\\
          & MCP    & 12.456 (0.024) & 14.83 (1.40) & 0.727 (0.090) &\~6.313 (0.617) & 27.2 (9.5)\~& 91.5 (1.7) & 0.0 (0.0)\\
          & SICA   & 12.438 (0.015) & 14.27 (1.24) & 0.733 (0.098) &\~5.768 (0.382) &  6.1 (0.3)  & 94.0 (0.0) & 0.0 (0.0)\\
          & Enet   & 12.630 (0.045) & 33.34 (3.47) & 2.026 (0.229) & 10.942 (0.689) & 57.0 (23.1) & 91.6 (2.7) & 0.6 (1.3)\\
          & Oracle & 12.703 (0.111) & 26.11 (3.88) & 1.219 (0.185) &\~9.679 (1.407) &100.0 (0.0)\~\~&\~0.0 (0.0) & 0.0 (0.0)\\
\addlinespace
$s=100$   & Lasso  & 13.652 (0.037) & 41.79 (2.54) & 2.514 (0.175) & 15.707 (0.605) & 31.2 (27.9) & 94.4 (4.3) & 2.6 (2.3)\\
$\ve=0.2$ & SCAD   & 13.581 (0.068) & 30.94 (7.51) & 1.659 (0.582) & 13.556 (1.437) & 60.6 (26.6) & 88.1 (5.7) & 1.0 (2.0)\\
          & MCP    & 13.578 (0.070) & 31.79 (7.33) & 1.738 (0.560) & 13.308 (1.570) & 30.2 (15.0) & 91.8 (4.3) & 1.1 (2.1)\\
          & SICA   & 13.523 (0.058) & 28.93 (4.88) & 1.569 (0.353) & 12.061 (1.313) &  6.5 (4.2)  & 94.4 (1.2) & 0.5 (1.4)\\
          & Enet   & 13.652 (0.038) & 41.80 (2.54) & 2.514 (0.175) & 15.744 (0.727) & 33.0 (30.4) & 94.4 (4.3) & 2.6 (2.2)\\
          & Oracle & 13.587 (0.121) & 27.41 (3.87) & 1.296 (0.201) & 10.307 (1.490) &100.0 (0.0)\~\~&\~0.0 (0.0) & 0.0 (0.0)\\
\midrule
\end{tabularx}
\end{sidewaystable}

\begin{table}
\centering
\caption{Results for various methods applied to the DLBCL data}\label{tab:dlbcl_rslt}\vskip1ex
\begin{tabular}{@{}lccc@{}}
\midrule
Method & \# of selected genes & Prediction error & $p$-value\\
\midrule
Lasso  & 24 & 0.3047 & 0.016\\
SCAD   & 21 & 0.3048 & 0.013\\
MCP    & 13 & 0.3083 & 0.012\\
SICA   & 11 & 0.3038 & 0.003\\
Enet   & 26 & 0.3046 & 0.013\\
\midrule
\end{tabular}
\end{table}

\begin{table}
\def\-{\phantom{-}}\centering
\caption{Estimated coefficients for selected genes in the DLBCL data by various methods}\label{tab:dlbcl_est}\vskip1ex
\begin{tabular}{@{}l*{5}{>{$}c<{$}}@{}}
\midrule
Gene ID & \multicolumn{1}{c}{Lasso} & \multicolumn{1}{c}{SCAD} & \multicolumn{1}{c}{MCP} & \multicolumn{1}{c}{SICA} & \multicolumn{1}{c@{}}{Enet}\\
\midrule
25054 &\-0.0066 &\-0.0064 &\-0.0035 &\-0.0102 &\-0.0065\\
33791 & -0.0041 & -0.0021 &         &         & -0.0048\\
27181 & -0.0090 & -0.0084 & -0.0019 &         & -0.0092\\
28654 &\-0.0039 &\-0.0035 &         &         &\-0.0037\\
31242 &\-0.0112 &\-0.0101 &\-0.0005 &\-0.0129 &\-0.0113\\
31981 &\-0.0223 &\-0.0200 &\-0.0222 &\-0.0410 &\-0.0214\\
27718 &\-0.0015 &\-0.0015 &         &         &\-0.0024\\
24725 &\-0.0095 &\-0.0085 &\-0.0133 &         &\-0.0092\\
27218 &\-0.0148 &\-0.0153 &\-0.0335 &\-0.0015 &\-0.0149\\
33014 &\-0.0067 &\-0.0054 &         &         &\-0.0075\\
16006 &\-0.0021 &         &         &         &\-0.0029\\
33974 &         &         &         &         &\-0.0003\\
27731 & -0.0252 & -0.0206 &         & -0.0551 & -0.0242\\
24394 & -0.0190 & -0.0198 & -0.0557 &         & -0.0197\\
24400 &\-0.0035 &\-0.0013 &\-0.0002 &         &\-0.0045\\
24203 &\-0.0002 &         &         &         &\-0.0012\\
24271 & -0.0065 & -0.0038 & -0.0001 & -0.0114 & -0.0069\\
25977 & -0.0197 & -0.0168 & -0.0157 & -0.0221 & -0.0200\\
34344 &\-0.0321 &\-0.0302 &\-0.0400 &\-0.0348 &\-0.0312\\
24530 &\-0.0004 &         &         &         &\-0.0011\\
27191 &         &         &         &         & -0.0002\\
33358 &\-0.0012 &\-0.0017 &         &         &\-0.0019\\
26470 &\-0.0067 &\-0.0036 &\-0.0048 &\-0.0074 &\-0.0063\\
26524 &\-0.0030 &\-0.0021 &         &         &\-0.0038\\
34376 &\-0.0278 &\-0.0243 &\-0.0311 &\-0.0296 &\-0.0273\\
32679 & -0.0075 & -0.0054 &         & -0.0053 & -0.0081\\
\midrule
\end{tabular}
\end{table}

\clearpage
\newgeometry{hmargin=0.75in,vmargin=1in}
\begin{figure}
\centering
\includegraphics[width=3.5in]{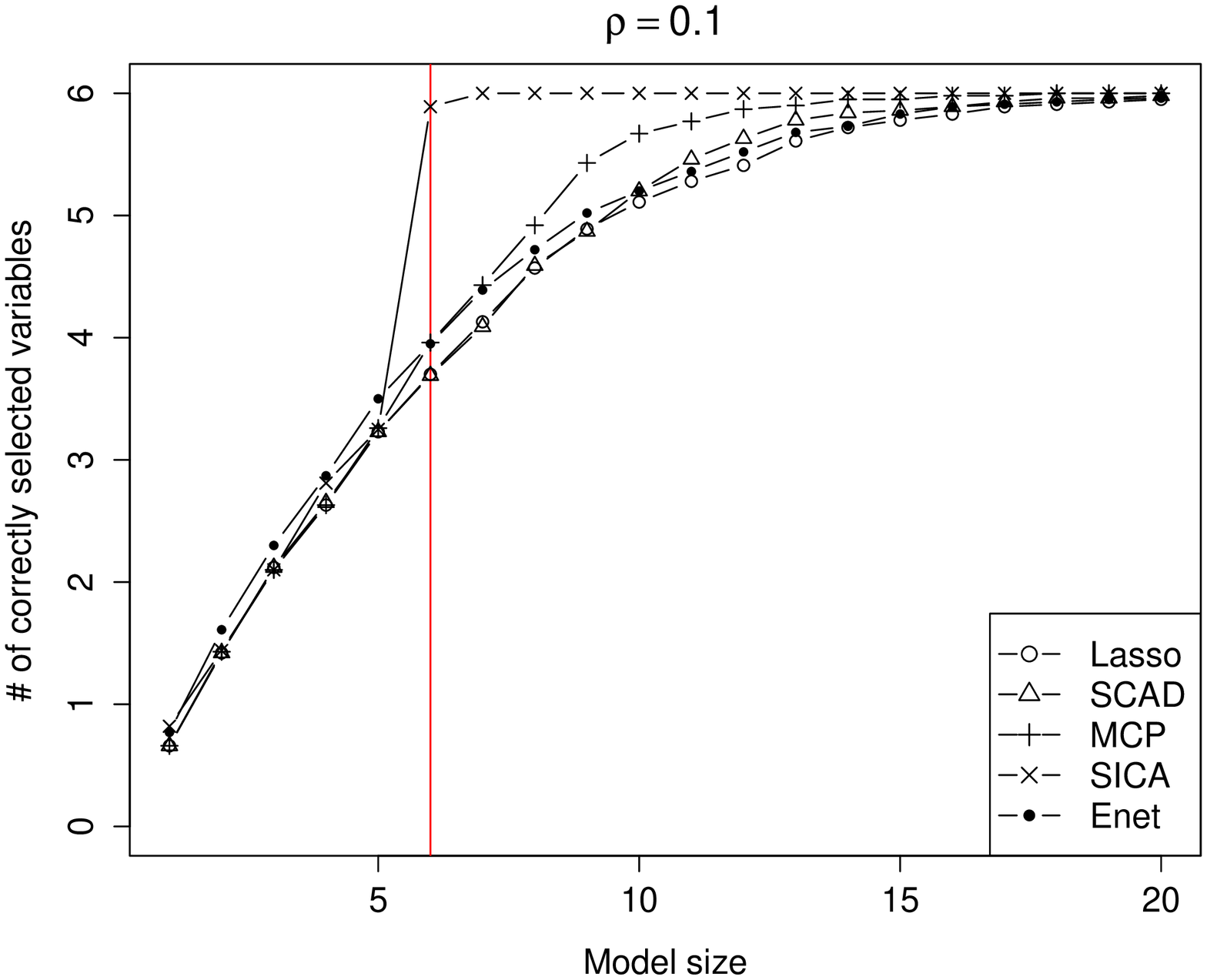}%
\includegraphics[width=3.5in]{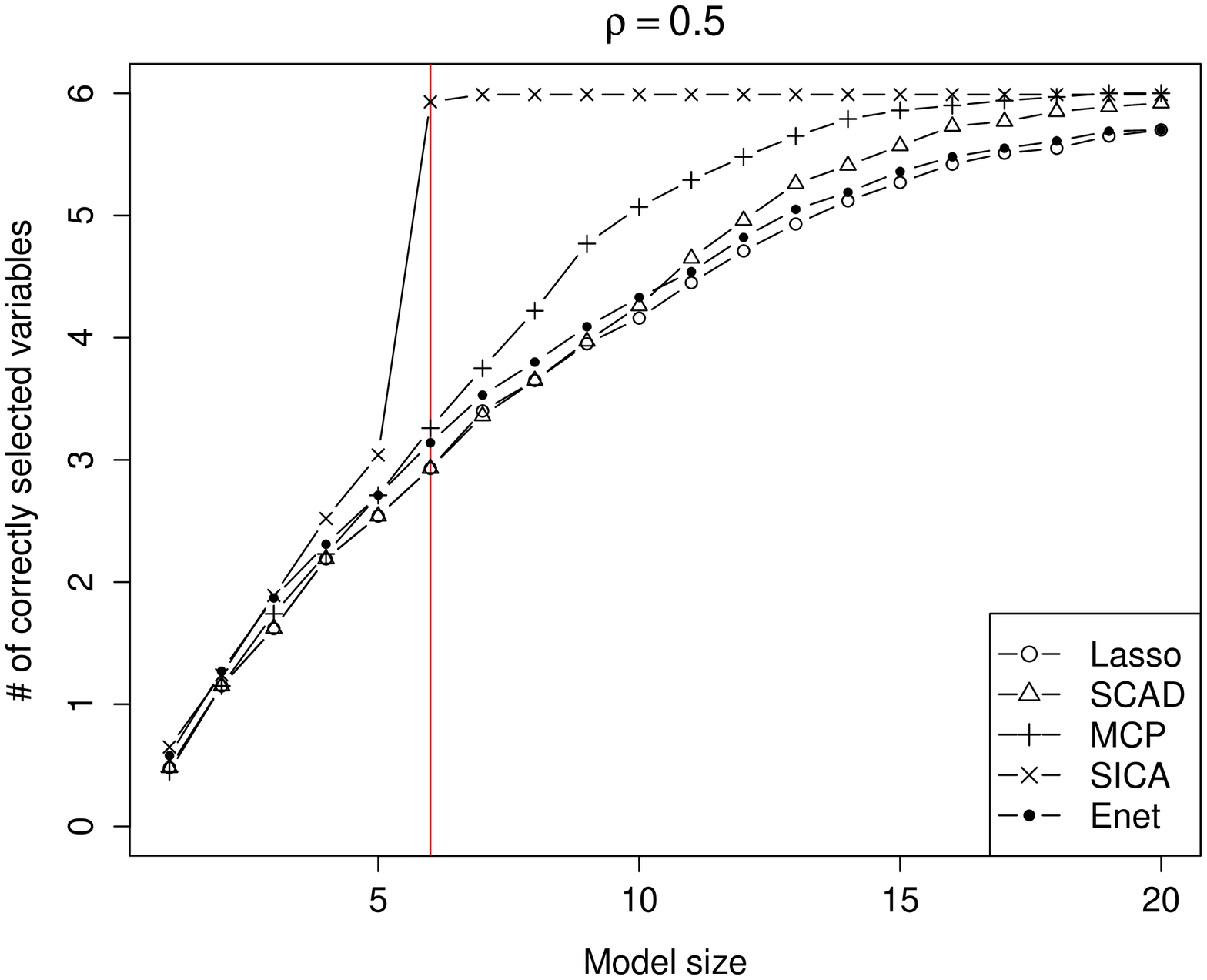}
\caption{Variable selection performance for various methods in the first simulation study with $p=100$. The vertical line indicates the true sparsity dimension.}\label{fig:low}
\end{figure}

\begin{figure}
\centering
\includegraphics[width=3.5in]{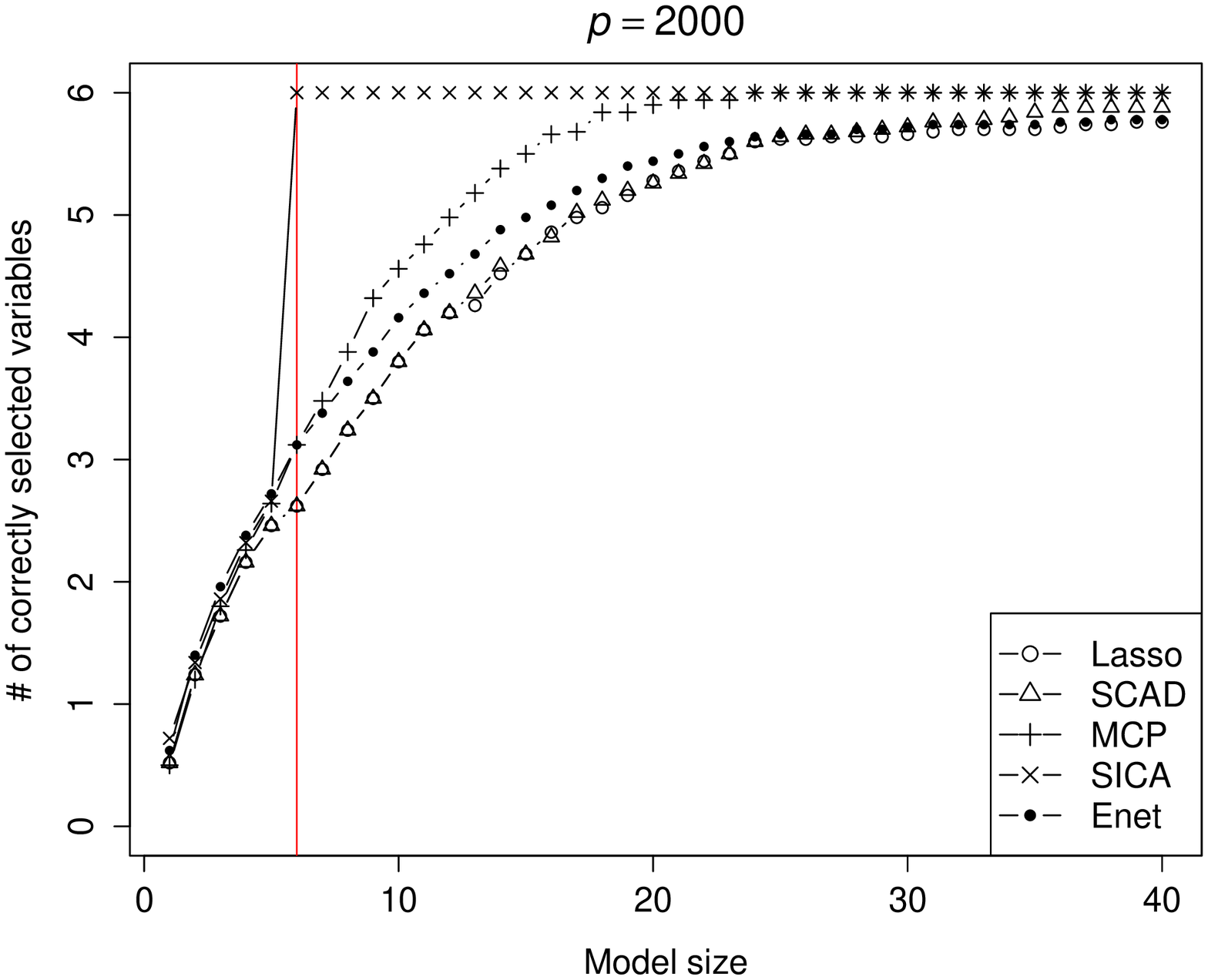}%
\includegraphics[width=3.5in]{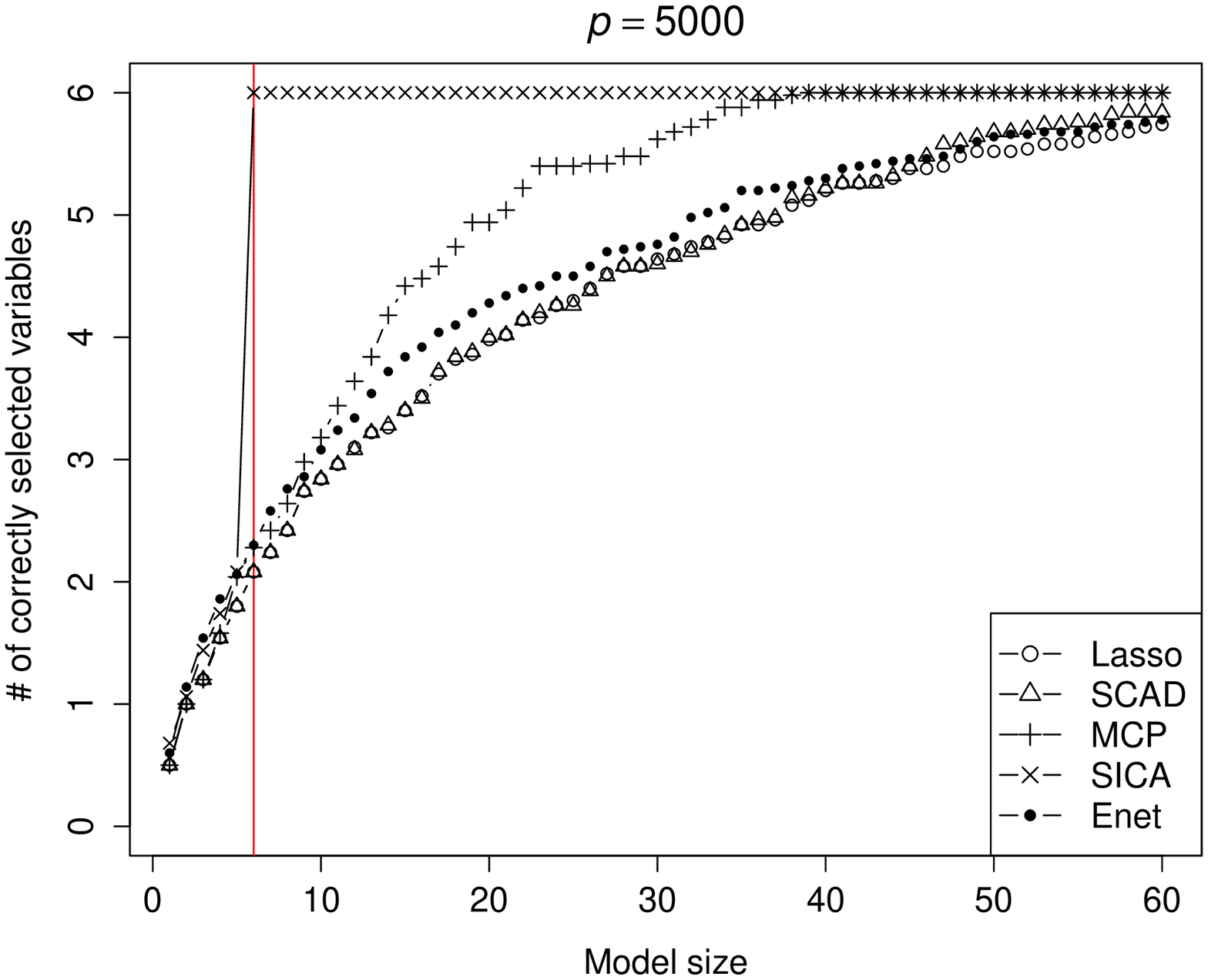}
\caption{Variable selection performance for various methods in the second simulation study with $\rho=0.5$. The vertical line indicates the true sparsity dimension.}\label{fig:high}
\end{figure}

\begin{figure}
\centering
\includegraphics[width=3.5in]{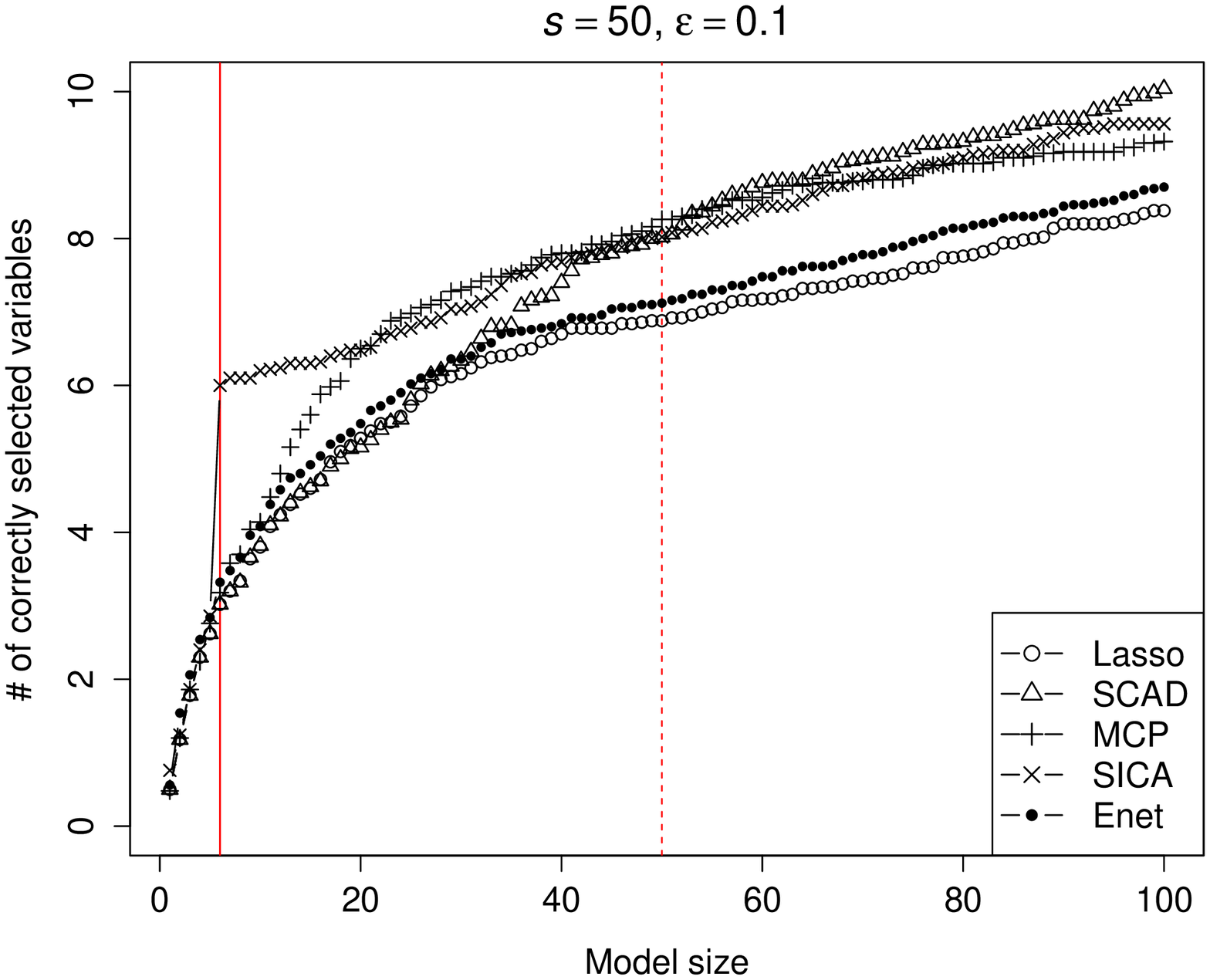}%
\includegraphics[width=3.5in]{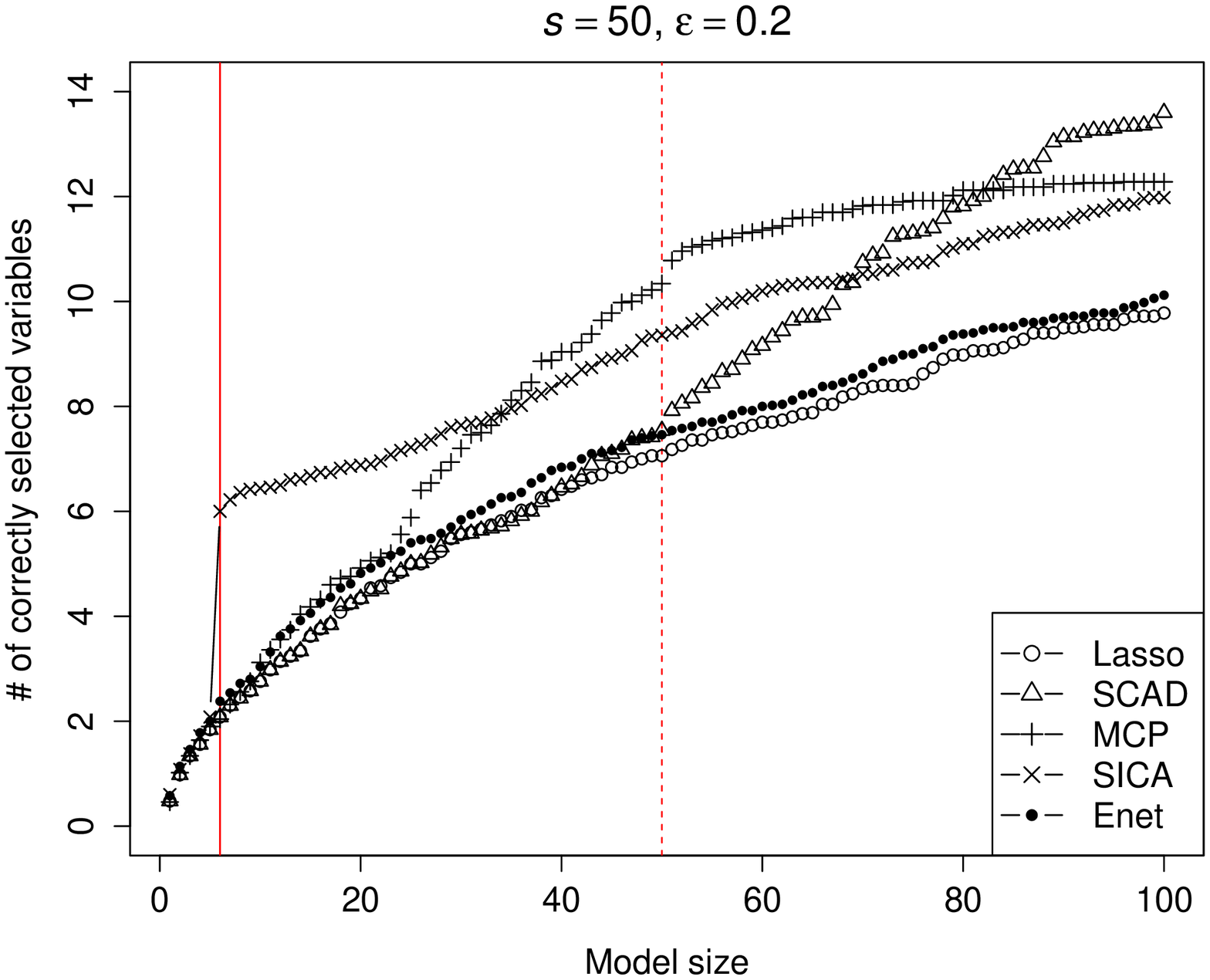}\\
\includegraphics[width=3.5in]{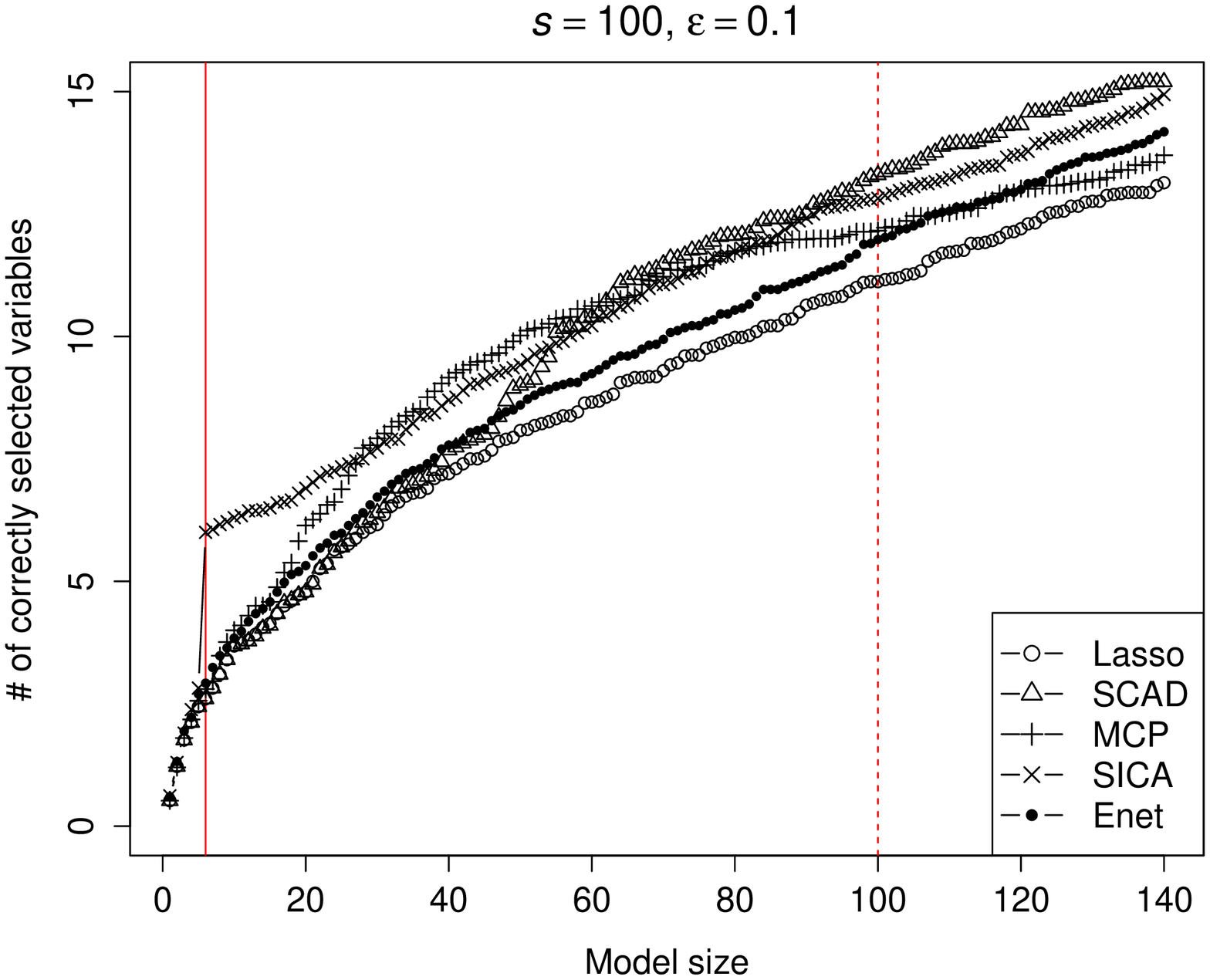}%
\includegraphics[width=3.5in]{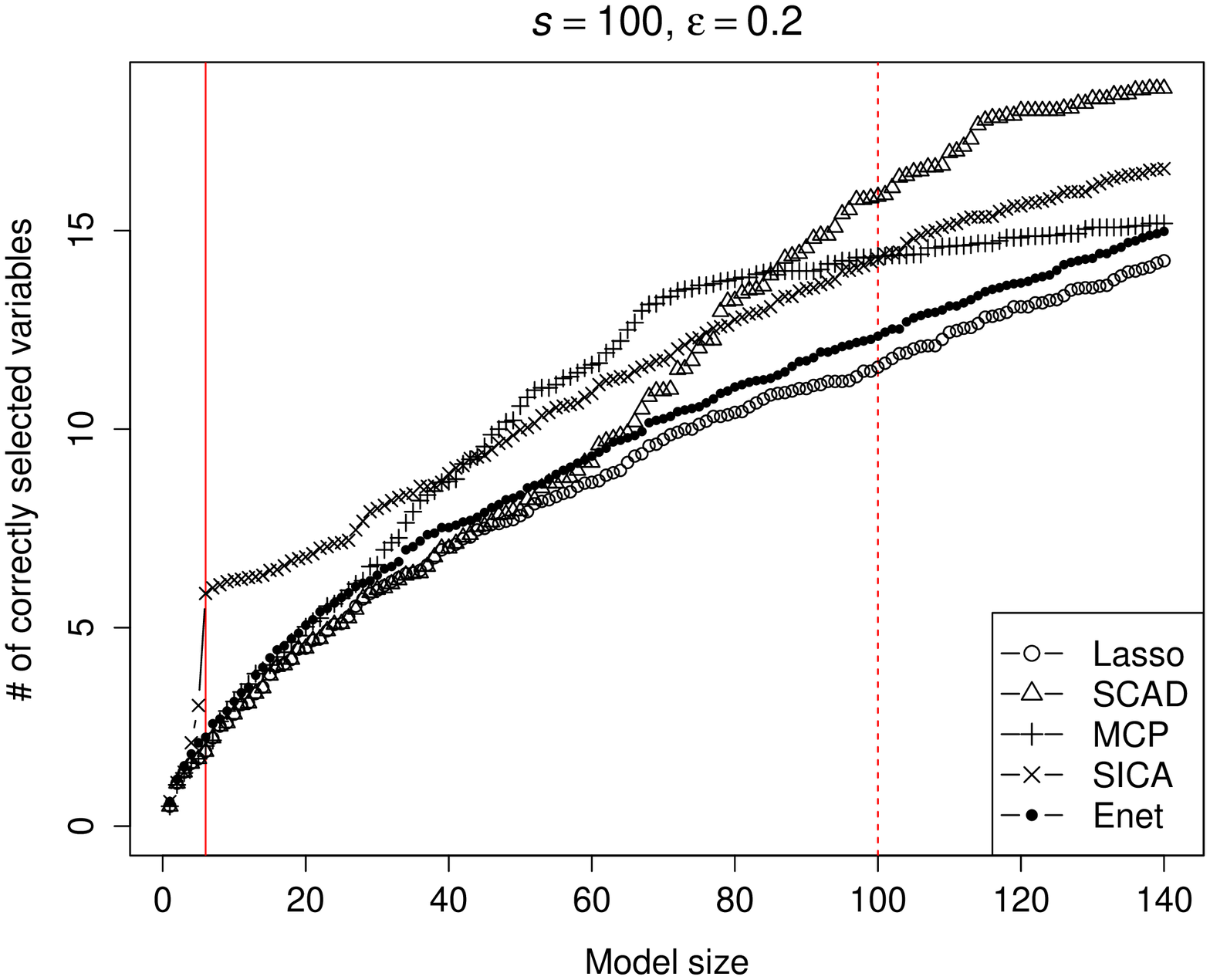}
\caption{Variable selection performance for various methods in the third simulation study. The vertical lines indicate the sparsity level with strong effects (solid) and the sparsity level with all nonzero effects (dashed).}\label{fig:hard}
\end{figure}

\end{document}